\documentclass{article}
\usepackage[utf8]{inputenc}

\usepackage[margin=1.2in]{geometry}
\usepackage{common}
\def\arxivversion{}

\title{Online multiple testing with e-values}
\author{
Ziyu Xu\thanks{Department of Statistics and Data Science, Carnegie Mellon University, USA. Email: \texttt{xzy@cmu.edu}.}
\and
Aaditya Ramdas\thanks{Departments of Statistics and Data Science, and Machine Learning, Carnegie Mellon University, USA. Email: \texttt{aramdas@cmu.edu}.}
}
\date{\today}

\begin{document}

\maketitle

\begin{abstract}
A scientist tests a continuous stream of hypotheses over time in the course of her investigation --- she does not test a predetermined, fixed number of hypotheses.
The scientist wishes to make as many discoveries as possible while ensuring the number of false discoveries is controlled --- a well recognized way for accomplishing this is to control the false discovery rate (FDR).
Prior methods for FDR control in the online setting have focused on formulating algorithms when specific dependency structures are assumed to exist between the test statistics of each hypothesis.
However, in practice, these dependencies often cannot be known beforehand or tested after the fact.
Our algorithm, e-LOND, provides FDR control under arbitrary, possibly unknown, dependence. 
We show that our method is more powerful than existing approaches to this problem through simulations. We also formulate extensions of this algorithm to utilize randomization for increased power, and for constructing confidence intervals in online selective inference.
 \end{abstract}

\tableofcontents

\section{Introduction}
Science advances one hypothesis at a time. Moreover, the rate at which new hypotheses are tested has drastically increased in recent decades to the point where a single scientist can quickly test hundreds to thousands of hypotheses with the aid of computation. For example, a geneticist can now sequence thousands of genes from trial subjects and individually determine whether each of these genes has an effect on phenotypes of interest (e.g., disease, physical characteristics, etc.). A team of data scientists can test many variations of a website or app in A/B experiments to determine which version maximizes desirable user metrics. The key feature of all these examples is that hypotheses are being formulated and tested in an \emph{online} fashion --- the total number of hypotheses that are tested is unknown beforehand and possibly infinite.
Thus, we can formulate the online multiple testing problem, as receiving a stream of hypotheses, \(H_1, H_2, \dots\) --- typically, these are the null hypotheses we wish to reject (e.g., this gene has no effect on this disease, there is no association between socioeconomic status and future earning potential, this recommendation algorithm does not increase average user view count, etc.). A subset, $\hypset_0 \subseteq \naturals$, of these null hypotheses are truly null, where $\naturals$ denotes the natural numbers. We wish to discover all the hypotheses that are not null, i.e., discover the non-null hypotheses $\hypset_1 \coloneqq \naturals \setminus \hypset_0$. For each hypothesis, we observe some data and must immediately decide whether it is a discovery or not before observing future hypotheses. Thus, we denote the set of discoveries so far as \(\rejset_1 \subseteq \rejset_2 \subseteq \dots \subseteq \naturals\). The \textit{false discovery proportion} (FDP) refers to the proportion of discoveries in a discovery set \(\rejset\) that are truly null. We want to control the \emph{false discovery rate} (FDR), which is the expectation of the FDP. Define these as follows. \begin{align}
    \FDP(\rejset) \coloneqq \frac{|\rejset \cap \hypset_0|}{|\rejset| \vee 1}, \qquad \FDR(\rejset) \coloneqq \expect[\FDP(\rejset)].
\end{align}
$(X_t)_{t \in \mathbb{I}}$ denotes a sequence of objects indexed by a set $\mathbb{I}$ --- we drop the index set and write $(X_t)$ it is clear from context (often $\naturals$). Our goal is to produce discovery sets $(\rejset_t)$ that satisfy the following guarantee:
\begin{align}
    \FDR(\rejset_t) \leq \alpha \text{ for all }t \in \naturals,
    \label{eq:online-fdr-control}
\end{align} while maximizing the number of discoveries.
FDR is reasonable metric to control in applications where one wishes to filter candidates that are promising before doing more extensive follow-up studies, e.g., clinical trials for drugs, genome-wide association studies for genetic factors, features for pushing to production, etc. We elaborate on the motivations for considering the FDR error metric in \Cref{sec:WhyFDR}.
\citet{robertson_online_multiple_2022a} comprehensively surveys the existing literature of online multiple testing.
In particular, multiple previous works have devoted significant effort to formulating different types of dependency that can arise in natural situations and deriving algorithms that provide online FDR control under these dependence structures \citep{zrnic_asynchronous_2021,zrnic2020power,fisher_online_control_2022,fisher_online_false_2022}.
These works have considered dependencies that are natural to the online setting (i.e., local dependence and dependence between asynchronously initiated experiments) as well as the popular PRDS condition \citep{benjamini_control_false_2001}. However, under \emph{unknown or arbitrary} dependence in the data, the assumptions for these algorithms are violated and they do not provably control the FDR.

There are many circumstances where one wishes to be robust to arbitrary dependence --- we list some below:
\ifarxiv{}{\vspace{-10pt}}
\begin{itemize}[leftmargin=*]
    \item \textbf{Data reuse}. A natural way in which unknown dependency might arise is when one uses the same dataset to evaluate a large number of hypotheses. Although reusing data for different hypothesis tests is not generally a statistically valid practice, this practice inevitably occurs, as data collection may be difficult or prohibitively expensive. For example, in many applied areas of machine learning, the same dataset may be used to evaluate many different methods, e.g., Kaggle competitions \citep{bojer_kaggle_forecasting_2021}, the UCI data repository~\citep{newman1998uci}. Similarly, open data repositories in science also are reused across many studies~\citep{auton_global_reference_2015,burton_genome-wide_association_2007,koscielny_international_mouse_2014}. Data reuse naturally comes up in offline policy evaluation in reinforcement learning, since often deploying a new policy has costs (e.g., expenses incurred by new actions, loss of revenue if a policy underperforms, etc.), and one would wish to backtest many policies on previously collected data. In all these cases, the statistics calculated for each test are highly dependent, since they use the same data.

\item \textbf{Temporal overlap}. This type of dependency is considered primarily in works involving local dependencies \citep{zrnic_asynchronous_2021}, as it occurs when data collected for different hypotheses overlap or are subject to temporal noise. For example, in A/B testing, users are incrementally added to each experiment over time. However, since there is no partitioning of users across experiments, experiments may overlap in users. This induces a dependence among the resulting test statistics. Temporal events (e.g., holidays or weekends) can also induce time-dependent noise. We elaborate on the ``doubly sequential framework'' relevant to this setting in \Cref{sec:Evalue}.

    \item \textbf{Inherent dependence}. Dependence between statistics might simply arise because of the data generating process. One common type is dependence that arises from sampling without replacement (WoR) from a finite population. Sampling WoR naturally arises when we wish to test the average treatment effect of a treatment on the finite population \citep{splawa-neyman_application_probability_1990} --- the statistics calculated for different treatments allocated to different samples are dependent --- we simulate our methods in this setting in \Cref{sec:Simulations}. Similarly, dependence also arises when doing coarser cluster (rather than individual) based randomization~\citep{campbell_developments_cluster_2007}. Dependence can also come from a data-dependent sampling mechanism, which we can observe in multi-armed bandits or adaptive sampling settings. \end{itemize}
In many experiments, one may not know ahead of time which combination of the aforementioned types of dependencies may occur, nor the specific structure they may take. This is particularly relevant in online multiple testing, since the nature of the hypotheses being tested and which types of data are being used to conduct the tests are not known a priori. Hence, being simultaneously powerful and robust to arbitrary dependence is a highly practical desiderata.

The primary of contribution of this paper is a new algorithm, \ELOND, that provably controls FDR, i.e., satisfies \eqref{eq:online-fdr-control}, under unknown and arbitrary dependence, while being more powerful (i.e, makes more discoveries) than previous state-of-the-art algorithms. Our method accomplishes this by utilizing \emph{e-values}, a class of statistics that has garnered significant recent attention in hypothesis testing. E-values are central in sequential testing \citep{ramdas_admissible_anytimevalid_2020,ramdas_how_can_2021} as every admissible sequential test utilizes an e-value. We characterize a ``doubly sequential framework'' of scientific experimentation that combines sequential tests with online multiple testing in \Cref{sec:Evalue}, and illustrate how retaining validity under arbitrary dependence is particularly useful in this framework. A notable example of an e-value is the universal inference statistic \citep{wasserman_universal_inference_2020}, which allows for testing of composite nulls without regularity conditions. This, in turn, enables the construction of tests for novel problems where no prior valid test exists --- an example of this is testing whether a distribution is log-concave \citep{dunn_universal_inference_2022,gangrade_sequential_test_2023}. The kinds of hypotheses for which e-values are applicable is quite comprehensive. We refer the reader to \citet{ramdas_gametheoretic_statistics_2022} for thorough collection of examples for which e-values are applicable.

\ifarxiv{\paragraph}{\noindent\textbf}{P-values vs.\ e-values.} Since the formulation of online multiple testing by \citet{foster2008alpha}, solutions have only assumed a \emph{p-value}, $P_t$, is associated with hypothesis $H_t$ and satisfies the following,
\begin{align}
\prob{P_t \leq s} \leq s \text{ for all }s \in [0, 1]\text{ if }t \in \hypset_, \label{eqn:PMarginalProb}
\end{align}
for all $t \in \naturals$. We consider the novel setting where, instead, an \emph{e-value}, $E_t$, accompanies each hypothesis $H_t$ and satisfies the following property for all $t \in \naturals$:
\begin{align}
\expect[E_t] \leq 1\text{ if }t \in \hypset_0.
    \label{eqn:EMarginalExpect}
\end{align}

An \emph{online multiple testing algorithm} is a sequence of (possibly random) test levels $(\alpha_t)$, where $\alpha_t \in [0, 1]$ for all $t \in \naturals$, and the algorithm produces discovery set $\rejset_t$ at the $t$th step in the following fashion:
\begin{align}
    \rejset_t  = \begin{cases}
        \{i \in [t]: P_i \leq \alpha_i\} \text{ if using p-values,}\\
        \{i \in [t]: E_i \geq 1 / \alpha_i\} \text{ if using e-values}\\
    \end{cases}.
\end{align}
The definition of $\rejset_t$ in the e-value case is equivalent to the p-value case if we assumed our p-values were formulated as $P_t = 1/  E_t$ --- one can see  this is a bona fide p-value by applying Markov's inequality to the e-value definition in \eqref{eqn:EMarginalExpect}.
One can consider e-value algorithms as operating on a special type of p-values. We leverage the specific properties of e-values to derive more powerful algorithms that remain valid even under arbitrary dependence.

\ifarxiv{\paragraph}{\noindent\textbf}{Our contributions.} We make the three following contributions in the main paper.
\ifarxiv{}{\vspace{-10pt}}
\begin{enumerate}[leftmargin=*]
    \item \emph{Powerful online FDR control under arbitrary dependence with e-values.} The current method for online FDR control under arbitrary dependence, the \rLOND\ algorithm \citep{javanmard2018online,zrnic_asynchronous_2021}, is unnecessarily conservative when applied to e-values. The \rLOND\ algorithm corrects each of its test levels by an additional factor that is logarithmic in the number of hypotheses tested so far, compared to its counterpart, the \LOND\ algorithm, that ensures FDR control under a much more stringent assumption of positive dependence. This is similar to the penalty paid by the Benjamini-Yekutieli procedure \citep{benjamini_control_false_2001} in the offline setting. Our algorithm, \ELOND, operates on e-values, but does not require the additional correction. Thus, it can maintain FDR control regardless of the dependence structure and dominates the standard \rLOND\ algorithm. Another previous approach to FDR control under dependence is the LORD$^*$ algorithm, which requires a priori knowledge of which hypotheses have statistics that are dependent. Our numerical simulations in \Cref{sec:Simulations} show that \ELOND\ is more powerful than \rLOND\, and becomes more powerful than LORD$^*$ when more hypotheses are mutually dependent.

    \item \emph{Additional power through randomization.} If one is interested in maximizing the power of their online multiple testing procedure, then one can incorporate randomization in the manner of of \citet{xu_more_powerful_2023}, who use randomization to improve offline multiple testing procedures. We develop variants of \ELOND\ and \rLOND\ (\ULOND\ and \UrLOND, respectively), that use the randomization of a single uniform random variable to increase their power over their deterministic counterparts.
These randomized methods dominate (i.e., never make fewer, and often make more discoveries) their deterministic versions and hence should be employed if one is interested in making as many discoveries as possible.

    \item \emph{Online FCR control with no restrictions on selection rules or dependence on e-CIs.} In addition to online FDR control, we also provide novel results for the online selective confidence interval (CI) problem introduced by \citet{weinstein2019online}. In this problem, one wishes to output, in an online fashion, CIs for a stream of parameters such that the overall false coverage rate (FCR) of all the CIs is controlled. This problem adds in the additional complexity of having a selection rule --- while a discovery is made at the $t$th hypothesis solely based on its test level $\alpha_t$, one decides whether a parameter should be selected for CI construction based on a selection rule $\Srm_t$ (which uses the observed data for the current and past parameters) that is separate from the coverage level of the CI, $1 - \alpha_t$. The extension of \ELOND\ to the online selective CI problem can control FCR under any sequence of selection rules, and arbitrary dependence. The sole caveat of this algorithm is that it operates on a subset of CIs based on e-values, called e-CIs \citep{vovk_confidence_discoveries_2023,xu_postselection_inference_2022}, which have been used for offline FCR control. 

\end{enumerate}

Our developments of \ELOND\ and \ULOND\ allow one to significantly improve power when e-values are available --- hence, our e-value methods are complementary to existing p-value based methods, i.e., \rLOND, for FDR control under arbtirary dependence. Our randomization techniques do benefit both e-value and p-value methods. Thus, a practitioner should use \rLOND\ or \UrLOND\ when only p-values are available, and \ELOND\ or \ULOND\ when e-values are available. When there is a mix of p-values and e-values, one should follow the guidance summarized in \Cref{corollary:LONDCalibration} of calibrating p-values to e-values.

\ifarxiv{}{\vspace{-10pt}}
\ifarxiv{\paragraph}{\noindent\textbf}{Outline.} In \Cref{sec:Evalue}, we discuss the ``doubly sequential'' framework that abstracts scientific experimentation. We recap existing online multiple testing algorithms and introduce the \ELOND\ algorithm in \Cref{sec:ELOND}.  In \Cref{sec:FCR} we devise methods for the online selective inference problem from \ELOND.
We demonstrate the power of \ELOND\ empirically through numerical simulations in  \Cref{sec:Simulations}, and summarize our findings in \Cref{sec:Conclusion}. We defer discussion of related work to \Cref{sec:RelatedWork,sec:ADA}. Further, we apply our methods to an online version of the model-free selective inference problem of \citet{jin_model-free_selective_2023} in \Cref{sec:conformal}, and evaluate their performance on real data from protein prediction task. Lastly, we show a sharpness result on the FDR control of \ELOND\ in \Cref{sec:sharp-fdr}, i.e., there exists instances where the true FDR is arbitrarily close to $\alpha$.

\section{Doubly sequential inference}
\label{sec:Evalue}
\begin{figure}[h]
    \centering
    \includegraphics[width=\ifarxiv{0.5\textwidth}{\columnwidth}]{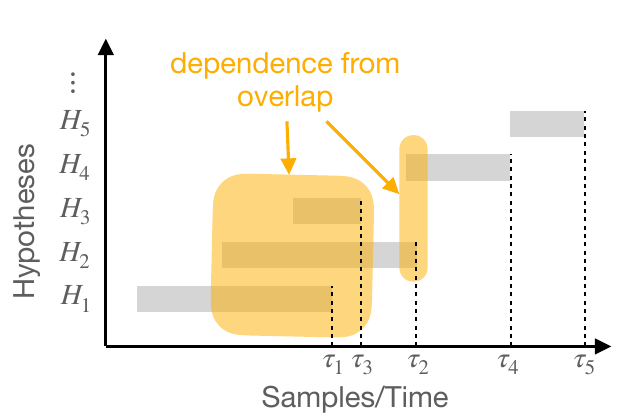}
    \ifarxiv{}{\vspace{-10pt}}
    \caption{A cartoon of the doubly-sequential framework for experimentation. Real world time or the number of samples collected is on the x-axis --- experiments run sequentially, stopping at $\tau_t$ when enough samples are collected for hypothesis $H_t$. Hypotheses arriving in a stream are shown on the y-axis. As a result of the overlap in time that data for different experiments is collected, dependence between hypotheses can occur.}
    \label{fig:DoubleSeq}
\end{figure}

E-values are particularly applicable to the sequential fashion in which data is gathered in many modern applications of hypothesis testing. In the sequential setting, samples are received one at a time, e.g., patients entering a clinical trial, users joining an A/B test, etc. To maximize efficiency, we collect data samples $X_1, X_2, \dots$ (here we are indexing by sample, rather than hypothesis) and stop sampling as soon as we are able to make a decision about the result of the experiment. A key concept for sequential testing is the \emph{e-process}, which is a process $(M_t)$ where $M_t$ is a function of the first $t$ samples $(X_1, \dots, X_t)$ and satisfies the following property:
\begin{align}
    \expect[M_\tau] \leq 1 \text{ for all stopping times }\tau\text{ under }H_0.\label{eq:e-process-def}
\end{align} A \emph{stopping time} is a random time $\tau$ that can be determined based on the data seen so far, i.e., one can determine whether $\tau = t$ solely by using $(X_1, \dots, X_t)$. From the definition of an e-process in \eqref{eq:e-process-def}, one can see  $M_\tau$ is an e-value, so making a discovery when $M_\tau \geq \alpha^{-1}$ is a valid hypothesis test with Type I error of at most $\alpha$. Consequently, a ubiquitous stopping time is the first time at which $M_t$ exceeds the test threshold $\alpha^{-1}$. \cite{ramdas_admissible_anytimevalid_2020} showed that any admissible sequential test which allows early stopping of this sort \emph{must} be derived from e-processes, making e-values a central and necessary component of sequential testing.

This leads us to the \emph{doubly sequential framework} \citep{robertson_online_multiple_2022a}, where both samples are hypotheses arrive sequentially, as a widely applicable framework for how scientific experimentation is done. An online multiple testing algorithm that utilizes e-values a quite useful in this framework, since e-values are critical to sequential testing.
\Cref{fig:DoubleSeq} illustrates this concept. Both data and hypotheses arrive in streams, and one must be able to test new hypotheses and utilize new data as evidence in an online fashion. When many experiments are run simultaneously, the data gathered for each experiment are \emph{dependent}, either due noise that jointly affects samples collected at a similar time (e.g., season fluctuations affecting the e-commerce habits of users), or because some experiments might share some of the collected data (e.g., clinical endpoints that utilize data from prior trials). Thus, a common application of this framework is in large scale A/B testing at companies \citep{xu2015infrastructure}, where separate data scientists are starting new experiments regularly, and have concurrent existing experiments that gather data sequentially.  \cite{yang_framework_multiarmedbandit_2017} illustrate an instance where the data for each hypothesis is collected through a multi-armed bandit.

Regardless, all these scenarios can involve complicated and unknown dependence between the statistics for testing each hypothesis. Thus, our methods that are robust to dependence allow for valid inference in the doubly sequential framework, and we verify this empirically in our experiments in \Cref{sec:Simulations}.

\section{e-LOND: FDR control via e-values}
\label{sec:ELOND}
To prepare ourselves for \ELOND,  we first recap what the current state-of-the-art algorithms are. Let  a \emph{discount sequence} $(\gamma_t)$ be a fixed sequence of nonnegative reals that satisfy $\sum_{t = 1}^\infty \gamma_t \leq 1$, and $\alpha \in [0, 1]$ is our desired level of FDR control.
For all sequences of discovery sets $(\rejset_t)$, we let $\rejset_0 = \emptyset$. An algorithm that produces a sequence of discovery sets $(\rejset_t^1)$ \emph{strictly dominates} an algorithm that produces $(\rejset_t^2)$ iff (1) $\rejset_t^1 \supseteq \rejset_t^2$ on all sequences of p-values $(P_t)$ (or e-values $(E_t)$) and all $t \in \naturals$, and (2) there is a sequence p-values $(P_t)$ (or e-values $(E_t)$) s.t.\ there exists $t \in \naturals$ where $\rejset_t^1 \supset \rejset_t^2$. Further, $(\rejset_t^1)$ is said to \emph{strictly dominate} $(\rejset_t^2)$ \emph{in expectation} if condition (1) holds and (3) if there also exists a sequence of p-values $(P_t)$ (or e-values $(E_t)$) and  $t \in \naturals$ such that $\expect[|\rejset_t^1| \mid (E_i)_{i \in [t]}] > \expect[|\rejset_t^2| \mid (E_i)_{i \in [t]}]$, i.e., the expected number of discoveries is strictly larger when taken only over the randomness in the algorithm.
\ifarxiv{\subsection{Prior work: the LOND and r-LOND algorithms}}{}
We first recall the \LOND\ algorithm. For each $t \in \naturals$ define:
\begin{align}
    \alpha_t^{\LOND} \coloneqq \alpha \gamma_t \cdot (|\rejset_{t - 1}^\LOND| + 1),
\end{align}
where $(\rejset_t^{\LOND})$ are the corresponding discovery sets. The \LOND\ algorithm requires p-values to be independent or positively dependent for FDR control.
\begin{fact}[Theorem 4 \citep{zrnic_asynchronous_2021}]
    For p-values $(P_t)$ that satisfy \eqref{eqn:PMarginalProb} and are independent or PRDS \citep[Definition 1]{zrnic_asynchronous_2021},
    $\FDR(\rejset_t^{\LOND}) \leq \alpha$ for each $t \in \naturals$.
    \label{fact:LONDIndPRDS}
\end{fact} To achieve FDR control under arbitrary dependence, the \rLOND\ algorithm outputs more conservative test levels. For each $t \in \naturals$, define
\begin{align}\label{eq:r-lond}
    \alpha_t^{\rLOND}\coloneqq \alpha \gamma_t \cdot \beta_t(|\rejset_{t - 1}^\rLOND| + 1).
\end{align} Here, $(\beta_t)$ is a sequence of \emph{reshaping functions} \citep{blanchard_two_simple_2008}. A reshaping function $\beta: [0, \infty) \mapsto [0, \infty)$ is a nondecreasing function that can be written in the form ${\beta(r) = \int_0^r x d\nu(x)}$ where $\nu$ is any probability measure on $[0, \infty)$.
Let $(\rejset_{t}^{\rLOND})$ denote the sequence of discovery sets output by \rLOND.
\begin{fact}[Theorem 2.7 \citep{javanmard_online_control_2015}, Theorem 4 \citep{zrnic_asynchronous_2021}
    \footnote{Strictly speaking, \rLOND\ in \citet{zrnic_asynchronous_2021} is formulated as $\alpha_t^{\rLOND} = \alpha\gamma_t \cdot \beta_t(|\rejset_{t - 1}| \vee 1)$ which is less powerful than~\eqref{eq:r-lond}, the latter being the original \rLOND\ \citep{javanmard_online_control_2015}. However, the proofs of~\cite{zrnic_asynchronous_2021} carry through to the original \rLOND.}
    ]\label{fact:LONDArbDep}
    Under arbitrary dependence in $(P_t)$, i.e., under \eqref{eqn:PMarginalProb},
    $\FDR(\rejset_t^{\rLOND}) \leq \alpha$ for each $t \in \naturals$.
\end{fact}
\ifarxiv{}{\vspace{-20pt}}
A typical choice of reshaping function is
\begin{align}
\beta_t^{\BY}(r) = (\lfloor r \rfloor \wedge t) / \ell_t,
\end{align}  where $\ell_t \coloneqq \sum_{i = 1}^t 1 / i$ --- this is the choice used by the Benjamini-Yekutieli (BY) procedure \citep{benjamini_control_false_2001} for offline FDR control.
Hence, one can consider \LOND\ and \rLOND\ as the online analogs of the Benjamini-Hochberg (BH) procedure \citep{benjamini_controlling_false_1995} for independent or PRDS p-values and the BY procedure for arbitrarily dependent p-values, respectively.
\ifarxiv{\subsection{The e-LOND algorithm}}{}
Our \emph{\ELOND\ algorithm} achieves the best-of-both worlds in the sense it has the same powerful test levels as \LOND, but also is valid under arbitrary dependence like \rLOND. For each $t \in \naturals$, define
\begin{align}
    \alpha_{t}^{\ELOND} \coloneqq \alpha \gamma_t \cdot (|\rejset_{t - 1}^\ELOND| + 1).
\end{align} $(\rejset_t^{\ELOND})$ denotes the resulting discovery sets. The following is our main result.
\begin{theorem}
    Under arbitrary dependence on e-values \eqref{eqn:EMarginalExpect}, \(\FDR(\rejset_t^\ELOND) \leq \alpha\) for each \(t \in \naturals\).
    In addition, \ELOND\ strictly dominates \rLOND\ applied to $(1 / E_t)$ for any sequence of reshaping functions $(\beta_t)$.
\label{thm:ELONDArbDep}
\end{theorem}
The proof relies on a simple observation about any e-value $E$ and test level $\alpha \in [0, 1]$ that allows us to directly upper bound the indicator of whether a discovery is made or not by the e-value itself:
\begin{align}
    \ind{E \geq \alpha^{-1}} = \ind{\alpha E \geq 1} \leq \alpha E.
    \label{eqn:DeterministicInequality}
\end{align} We defer the full proof to \Cref{sec:ArbDepProof}. Further, we show in \Cref{sec:sharp-fdr} that this level of FDR control is \emph{sharp}, i.e., one can design instances of e-values where the true FDR of \ELOND\ is arbitrarily close to the upper bound of $\alpha$.

The \ELOND\ algorithm has the same test levels $(\alpha_t)$ as \LOND, but we use different notation to emphasize that \ELOND\ operates on e-values with no restrictions on dependence and \LOND\ operate on p-values that are independent or satisfy PRDS. This is similar to the relationship between the e-BH procedure \citep{wang_false_discovery_2020} and BH for offline FDR control.

In addition, we can show \rLOND\ is actually a special case of \ELOND. To clarify how \rLOND\ is subsumed by \ELOND\ under arbitrary dependence, we introduce the notion of calibration. Any p-value $P$ can be calibrated into an e-value $E = f(P)$ using a \emph{ calibrator} \citep{vovk_e-values_calibration_2021}. A calibrator ${f: [0, 1] \mapsto [0, \infty)}$ is an nonincreasing, upper semicontinuous function that satisfies $\int_{0}^1 f(x) dx \leq 1$. We can define a specific sequence of calibrators $(f_t)$ that transform p-values into e-values such that \rLOND\ is a special case of \ELOND.
\begin{corollary}
If p-values $(P_t)$ satisfy \eqref{eqn:PMarginalProb}, we can construct an e-value $E_t = f_t(P_t)$ for each $t \in \naturals$
from a sequence of calibrators $(f_t)$. We achieve $\FDR(\rejset_t^\ELOND) \leq \alpha$ for each $t \in \naturals$ by \Cref{thm:ELONDArbDep}. If we define $f_t$  as follows:
\begin{align}
    f_t(p) =  (\alpha \gamma_t \cdot \lceil (p \ell_t /(\alpha \gamma_t)) \vee 1 \rceil)^{-1}
    \label{eqn:LONDCalibration}
\end{align} we recover \rLOND\ for FDR control under arbitrary dependence described in \Cref{fact:LONDArbDep}. This allows us to reap the benefits of \ELOND\ when only some hypotheses may have e-values, and the rest have p-values ---  we can calibrate just the p-values before running \ELOND.
\label{corollary:LONDCalibration}
\ifarxiv{}{\vspace{-20pt}}
\end{corollary}

\ifarxiv{\subsection}{\paragraph}{More power through randomization}  Building on recent advances by \citet{xu_more_powerful_2023} for offline multiple testing, we can strictly improve both \ELOND\ and \rLOND\ by incorporating independent randomization.
Let $E$ be an e-value and $\widehat\alpha \in [0, 1]$ be a possibly random threshold that may depend on $E$.
Let $U$ be a uniform random variable on $[0, 1]$ that is independent of both $E$ and $\widehat\alpha$. Define the following randomized e-value:
\begin{align}
    S_{\widehat\alpha}(E) \coloneqq (E \cdot \ind{E \geq \widehat\alpha^{-1}}) \vee (\ind{U \leq E \widehat\alpha}\widehat\alpha^{-1}),
\end{align}
\begin{fact}[Proposition 2 \citep{xu_more_powerful_2023}]\label{fact:stochastic-rounding-e}
    $S_{\widehat{\alpha}}(E)$ is also an e-value. Further, note that
\begin{align}
    \ind{S_{\widehat\alpha}(E) \geq \widehat{\alpha}^{-1}} = \ind{E \geq \widehat{\alpha}^{-1} \cdot U}
\end{align}
\end{fact}
We now define \ULOND, a randomized version of \ELOND. Let $(U_t)$ be a sequence of  uniform random variables on $[0, 1]$ that are independent of $(E_t)$.
\begin{align}
    \alpha^{\ULOND}_t \coloneqq \alpha_t^{\ELOND}\cdot U_t^{-1}.
    \label{eq:ULOND}
\end{align} Let $(\rejset_t^{\ULOND})$ be the sequence of discovery sets output by \ULOND. The following is our second main result.
\begin{theorem}
    Under arbitrary dependence on e-values \eqref{eqn:EMarginalExpect}, $\FDR(\rejset_t^\ULOND) \leq \alpha$ for each $t \in \naturals$. Further, \ULOND\ strictly dominates \ELOND in expectation.\end{theorem}
\ifarxiv{}{\vspace{-10pt}}
\begin{proof}
    \ULOND\ in \eqref{eq:ULOND} is equivalent to applying \ULOND\ to $(S_{\alpha_t^{\ELOND}}(E_t))$. Hence, FDR control holds by \Cref{thm:ELONDArbDep}. The domination is because $U_t^{-1} > 1 + \varepsilon$ with nonzero probability for all $\varepsilon > 0$, and is independent from $(E_t)$.
\end{proof}
Note that $(U_t)$ can all be equal, i.e., $U_1 = \dots = U_t$, or they can be drawn independently for each hypothesis.
To improve \rLOND, we use the following result.
\begin{fact}[Lemma 1 \citep{xu_more_powerful_2023}]
Let $P$ be a superuniform random variable that can be arbitrarily dependent on a positive random variable $R$. Let $U$ be a superuniform
random variable that is independent of both $P$ and $R$. Let $c$ be a nonnegative constant and
$\beta$ be a reshaping function. Then, the following holds:
\begin{align}
    \expect\left[\frac{\ind{P\leq c\beta(R / U)}}{R}\right] \leq c.
\end{align}
\label{fact:rand-su-lemma}
\end{fact}
\ifarxiv{}{\vspace{-20pt}}
We can define the \UrLOND\ procedure as follows:
\begin{align}
    \alpha_t^{\UrLOND} \coloneqq \alpha \gamma_t \beta_t((|\rejset_{t - 1}| + 1) / U_t),
    \label{eq:UrLOND}
\end{align} with $(\rejset_t^{\UrLOND})$ being the resulting discovery sets.
We now present our third main result.
\begin{theorem}\label{thm:UrLOND}
    Under arbitrary dependence on p-values \eqref{eqn:PMarginalProb}, $\FDR(\rejset_t^\UrLOND) \leq \alpha$ for each $t \in \naturals$. Further, \UrLOND\ strictly dominates \rLOND\ in expectation for reshaping functions $(\beta_t^{\BY})$.
\end{theorem}
\ifarxiv{}{\vspace{-5pt}}
We defer the proof to \Cref{sec:urlond-proof}.
\begin{corollary}
    If we use reshaping function $\beta_t^{\BY}$, \UrLOND\ produces the following test levels:
    \begin{align}
        \alpha_t^{\UrLOND} = \alpha \gamma_t (\lfloor (|\rejset_{t - 1}^\UrLOND| + 1) / U\rfloor \wedge t) /\ell_t.
    \end{align}
\end{corollary}
\ifarxiv{}{\vspace{-10pt}}
Thus, by utilizing randomization, we are able to derive FDR controlling procedures that are never worse than their deterministic counterparts.

\section{Online FCR control with e-CIs}\label{sec:FCR}

Often, a scientist wishes not only to test the significance of an effect but also to measure the strength of the effect. Instead of receiving hypotheses in a stream, a scientist can consider a stream of parameters $\theta_1 \in \Theta_1, \theta_2 \in \Theta_2, \dots$, but wishes to estimate only some of them, e.g., only ones that show signficiant positive effect. Here, we desire our selected CIs to be accurate in aggregate, i.e., we want to control the false coverage rate (FCR) --- this problem was introduced by \cite{weinstein2019online} as the \emph{the online selective-CI problem}. For the $t$th parameter, the scientist receives some data (e.g., the results of an experiment) $X_t \in \Xcal_t$ and designs a selection rule ${\Srm_t: \Xcal_t \mapsto \{0, 1\}}$ to decide whether CI should be constructed for $\theta_t$. If a parameter is selected, one must choose an error level $\alpha_t \in (0, 1)$ and construct a $(1 - \alpha_t)$-CI for $\theta_t$. Let ${S_t = \Srm_t(X_t)}$ be an indicator variable that is 1 iff $\theta_t$ is selected for CI construction. We assume that one has access to a CI constructor $C_t: \Xcal_t \times [0, 1] \mapsto 2^{\Theta_t}$ for each $t \in \mathbb{N}$ where $C_t(X, \alpha)$  satisfies the following property:
\begin{align}
    \prob{\theta_t \not\in C_t(X_t, \alpha)} \leq \alpha \text{ for every }\alpha \in [0, 1].\label{eq:ci-def}
\end{align}
Formally, the \emph{false coverage proportion} (FCP), and the \emph{false coverage rate} (FCR) are defined as follows:
\begin{gather}
    \FCP(\Scal_t) \coloneqq \sum\limits_{i \in \Scal_t}\frac{ \ind{\theta_i \not\in C_i(\alpha_i)}}{|\Scal_t| \vee 1},\\ \FCR(\Scal_t)\coloneqq \expect\left[\FCP(\Scal_t)\right].
\end{gather}
The methods of \cite{weinstein2019online} relied on two key assumptions. The first is an explicit assumption on the dependence between hypotheses, i.e.,  $X_t$ were independent or that $C_t(X_t, \alpha_t)$ is still a valid $(1 - \alpha_t)$-CI conditional on past selection decisions. The second is a restrictive monotonicity assumption on the selection rules $\Srm_t$. In \Cref{alg:ELONDCI}, we devise versions of \ELOND\ and \ULOND\ for the online selective inference problem, \ELOND-CI and \ULOND-CI, respectively, that is free of both restrictons.

\begin{algorithm}[h!]
    \caption{The \ELOND-CI and \ULOND-CI algorithms ensure $\FCR \leq \alpha$ with no restrictions on the dependence between data $(X_t)$ or the selection rules $(\Srm_t)$. Let $(U_t)$ uniform random variables on $[0, 1]$ and independent of $(X_t)$.}
    \begin{algorithmic}
        \label{alg:ELONDCI}

        \STATE \textbf{Input:} E-CI constructors $(C_t)$, discount sequence $(\gamma_t)$, and FCR control level $\alpha$.
        \FOR{each $t \in \naturals$}
            \IF{running \ELOND-CI}
                \STATE $\alpha_t\coloneqq \alpha \gamma_t (|\Scal_{t - 1}| + 1)$.
            \ELSIF{running \ULOND-CI}
                \STATE $\alpha_t\coloneqq \alpha \gamma_t (|\Scal_{t - 1}| + 1) \cdot U_t^{-1}$.
            \ENDIF
            \STATE Receive data $X_t$.
            \STATE Make a selection decision $S_t \coloneqq \Srm_t(X_t)$.
            \IF{$S_t = 1$}
                \STATE $\Scal_t \coloneqq \Scal_{t - 1} \cup \{t\}$.
                \STATE Construct $C_t(X_t, \alpha_t)$ for $\theta_t$.
            \ELSE{}
                \STATE $\Scal_t \coloneqq \Scal_{t - 1}$
            \ENDIF
        \ENDFOR
    \end{algorithmic}
\end{algorithm}
To ensure FCR control, both algorithms do require each $C_t$ to a special type of CI: an \emph{e-CI} \citep{vovk_confidence_discoveries_2023,xu_postselection_inference_2022} --- similar to how \ELOND\ applies to e-values. $C(X, \alpha)$ is an e-CI over the universe of parameters $\Theta$ if it can be written as follows:
\begin{align}
    C(X, \alpha) = \{\theta \in \Theta: E_\theta < \alpha^{-1}\}, \label{eq:eci-def}
\end{align} where $E_\theta$ is an e-value when the true parameter is $\theta$. Note that the e-CI in \eqref{eq:eci-def} does satisfy the CI definition in \eqref{eq:ci-def} by Markov's inequality applied to $E_{\theta^*}$, where $\theta^*$ is the true parameter. Let $(\Scal^{\ELOND}_t)$ and $(\Scal^{\ULOND})$ denote the resulting selection sets of \ELOND-CI and \ULOND-CI, respectively. We now present our fourth main result, whose proof is in \Cref{sec:FCRProof}.
\begin{theorem}\label{thm:FCR}
    For any dependence structure among the data, $(X_t)$, and sequence of selection rules $(\Srm_t)$, $\FCR(\Scal_t^\ELOND), \FCR(\Scal_t^\ULOND) \leq \alpha$ for all $t \in \naturals$.
\end{theorem}
\begin{remark}
    Unlike discovery sets $(\rejset_t)$ in the the online FDR control problem, the selection sets $(\Scal_t)$ \emph{do} not depend on $(\alpha_t)$ --- $(\Scal_t)$ can be chosen in an arbitrary fashion based on the observed data. Thus, algorithms with online FDR control do not necessarily provide provide FCR control. However, the reverse is true --- FCR control implies FDR control \citep[Section 5.2]{weinstein2019online}.
\end{remark}
As discussed by \citet{xu_postselection_inference_2022}, many existing canonical CIs are e-CIs, in the same way that many p-values are implicitly inverted e-values. This gives \ELOND-CI and \ULOND-CI broad applicability and utility as a default online selective inference method that is robust to the unknown dependence and arbitrary user choice of selection rule. 

\section{Numerical simulations}\label{sec:Simulations}

To highlight the practical behavior of our methods, we conduct two simulations, with different dependence structures, where we test the null hypothesis $H_0: \mu \leq 0$, where $\mu$ is the mean of a distribution with support bounded in $[-4, 4]$. The first simulation is with local dependence between hypotheses, and the second is with sampling without replacement (WoR) dependence between hypotheses. In both instances, it we sample data sequentially, and hence our experiments exemplify the practicality of our new e-value based methods for the doubly sequential framework described in \Cref{sec:Evalue}. In addition to simulations, we also describe an application of our methods to \emph{online model-free selective inference under covariate shift} in \Cref{sec:conformal}, and compare the performance of our methods on real data from a protein prediction task from \cite{jin_model-free_selective_2023}.

\begin{figure*}[tbh]
    \centering
    \includegraphics[width=\textwidth]{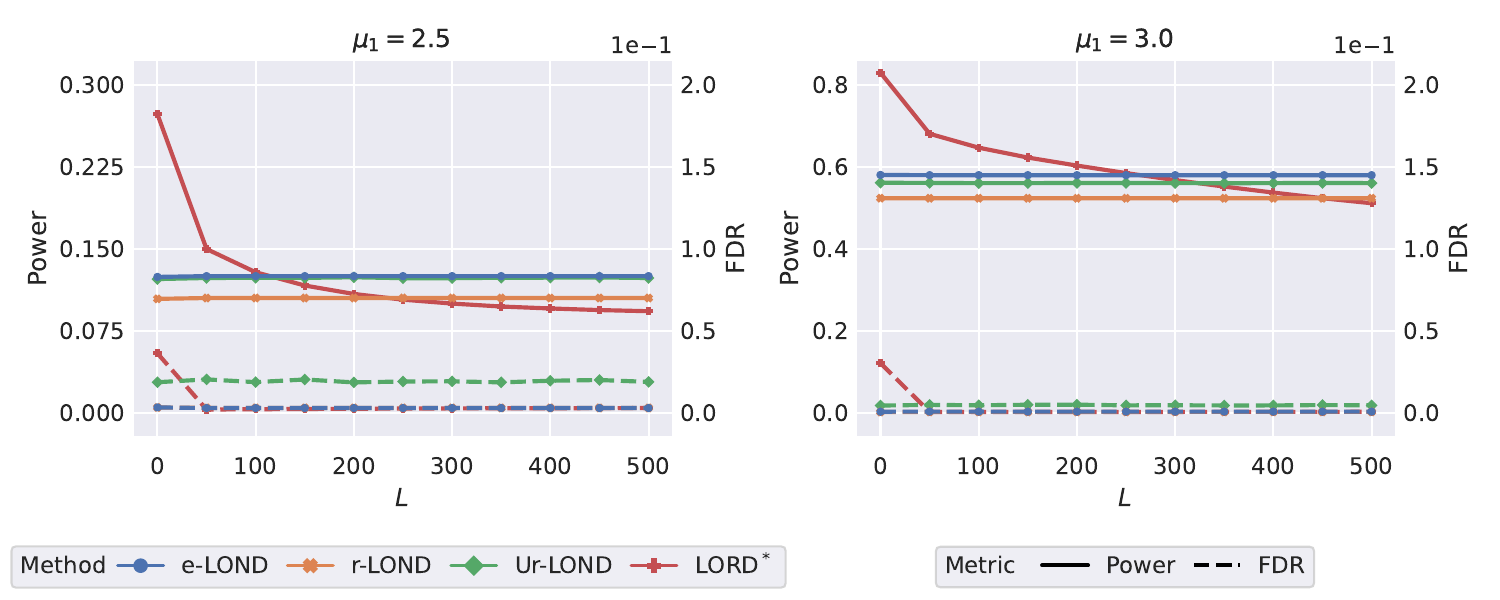}
    \caption{The power of different methods with provable FDR control against the lag parameter $L$ in a simulation with local dependence between statistics. Empirically, the FDR of all methods is well below the desired level of $\alpha=0.3$. As $L$ increases (i.e., more hypotheses are dependent), we can see the power of \LORD$^*$ decrease, since it is essentially ignoring hypotheses with statistics that are dependent with the current hypothesis being tested. \ELOND\ has consistently higher power than both p-value procedures, \rLOND\ and \UrLOND, and has higher power than LORD$^*$ as $L$ becomes large. We omit \ULOND\ since its power increase over \ELOND\ is very small. All Monte Carlo error from simulations is negligible (smaller than the line width in the plot).}
    \label{fig:LagCompSimulation}
    \ifarxiv{}{\vspace{-10pt}}
\end{figure*}
\ifarxiv{\subsection{Local dependence}}{\noindent\textbf{Local dependence.}}
We perform numerical simulations comparing \ELOND\ to other methods in a version of the local dependence setting from \citet{zrnic_asynchronous_2021}. Here, we draw data in a sequential setting with bounded random variables, since powerful sequential p-values for testing the mean of bounded random variables are naturally derived from e-values. We let $L$ be our local dependence lag parameter, i.e., the data for the $t$th hypothesis is independent of data from hypotheses that are more than $L$ indices away.
We let the total number of hypotheses be $T = 10^3$. For the $t$th hypothesis, we consider a setup where we recieve stream of $N = 200$ samples $(X_t^i)_{i \in [N]}$, where $X_t^i$ for each $i \in [N]$ are sampled i.i.d.\ from a Beta distribution (shifted and scaled to be on $[\pm 4]$) with mean $\mu_0 = 0$ under the null, and $\mu_1 \in \{2.5, 3\}$ otherwise.
For each $i \in [N], t \in [T]$, $X_t^i$ has Gaussian copula dependence with $(X^i_{t - L}, \dots, X^i_{t + L})$, i.e, the $i$th sample of data for hypotheses that are within $L$ steps. Explicitly, the covariance matrix of the Gaussian distribution, $\Sigma$, is set to $\Sigma_{i, j} = 0.5^{|i - j|}$ when $|i - j| \leq L$ and 0 elsewhere.
We construct p-values and e-values that are valid for this setting based on Hoeffding's inequality (see \Cref{sec:local-dep-details} for details).

Our results are averaged over 500 trials and shown in \Cref{fig:LagCompSimulation}. In addition to comparing to \rLOND\ and \UrLOND, we compare to \LORD$^*$, which is online FDR control algorithm from \citet{zrnic_asynchronous_2021} requires knowing the lag parameter $L$ beforehand, so it can solely utilize test statistics from hypotheses that are independent from the current hypothesis (see \Cref{sec:lord-star} for details). The power of $\LORD^*$ degrades as the lag parameter increases, which is expected, since it has access to a decreasing number of discoveries.
\ELOND\ is more powerful than both \rLOND\ and \UrLOND\ across the board, and $\LORD^*$ once $L \geq 250$ ($\mu_1=3$) or $L \geq 150$ ($\mu_1 = 2$). \ULOND\ only offers a small increase in power over \ELOND\ here so it is omitted.

\begin{figure*}[h]\label{fig:wor-simulation}
    \centering
    \includegraphics[width=\textwidth]{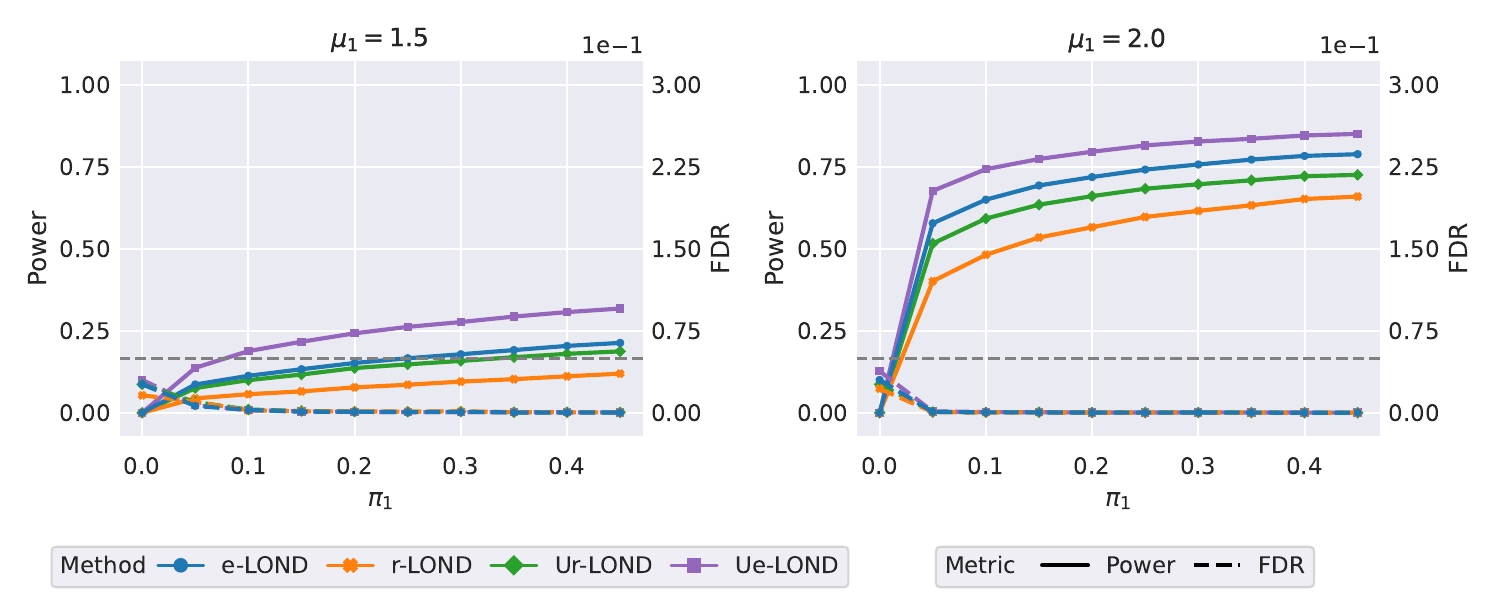}
    \caption{The power of different methods with provable FDR control against proportion of non-nulls $\pi_1$ in a simulation with sampling without replacement (WoR) dependence between statistics. Empirically, the FDR of all methods are below $\alpha=0.05$. \ELOND\ has consistently higher power than both p-value procedures, \rLOND\ and \UrLOND, and \ULOND\ is consistenly more powerful than \ELOND. This makes two e-value procedures, the most powerful methods. All Monte Carlo error from simulations is negligible (smaller than the line width in the plot).}
    \ifarxiv{}{\vspace{-10pt}}
\end{figure*}
\ifarxiv{\subsection{Sampling WoR}}{\noindent\textbf{Sampling WoR.}} We construct a population such that the mean is $\mu_0 = 0$ for the data we sample WoR for the null hypotheses, positive $\mu_1$ for the non-null hypotheses.
We will construct this population by discretizing a scaled and shifted Beta distribution.
Let $V(0), V(1) \in [\pm 4]^{N \times T}$ be the populations created from $P(\mu_0), P(\mu_1)$. Let $V_{i, t}$ be the $t$th value in $V(i)$. We set $s=0.01, \mu_0 = 0, \mu_1 \in \{1.5, 2\}$ in our simulations.
For each simulation trial, we choose a non-null proportion $\pi_1 \in [0.1, 0.9]$, and uniformly randomly choose $B \in \{0, 1\}^T$ with exactly $\lceil \pi_1 T \rceil$ ones. Let $\sigma$  be a random permutation over $[N \times T]$.
Our data for the $t$th hypothesis is $X_t = (V_{B_t, \sigma((t - 1) \cdot N + i)})_{i \in [N]}$. $X_t$ is a sample WoR of size $N$ from $V(0)$ if $B_t = 0$ and $V(1)$ if $B_t = 1$. Our e-values and p-values usiang an e-process for sampling WoR from~\cite{waudby-smith_confidence_sequences_2020} --- see \Cref{sec:sampling-wor-details} for details.

Our results, averaged over 500 trials, are in \Cref{fig:wor-simulation}. Here, both \ELOND\ and \ULOND\ dominate in power across the board, while all methods have FDR below $\alpha=0.05$. Clearly, the theoretical improvements of our novel e-value methods translate to empirical gains.

 \section{Application: online model-free selective inference under covariate shift}
\label{sec:conformal}

As an application of our framework, we can address an online version of the model-free selective inference under covariate shift problem introduced by \citet{jin_model-free_selective_2023}. To do so, we use \ELOND\ to directly derive an online version of the \emph{weighted conformal selection} (WCS) procedure. In this setting, we consider labeled pairs $(X_i, Y_i) \in \Xcal \times \Ycal$. We are given an i.i.d. calibration dataset of labeled pairs $\{(X_i, Y_i)\}_{i\in [n]}$ where $(X_i, Y_i) \sim \mathbf{P}$.
Our goal is to perform inference on a stream of i.i.d.\ test data points $(X_{n + 1}, Y_{n + 1}), (X_{n + 2}, Y_{n + 2}), \dots$. For each $t \in \naturals$, we only observe the covariates of the test points, $X_{n + t}$, and a potentially random threshold, $c_{n + t}$. Our goal is to test the following hypothesis about $Y_{n + t}$:
\begin{align}
H_0^t: Y_{n + t} \leq c_{n + t}.
\end{align}
One notable difference between the setup here and standard online multiple testing is that the null hypotheses themselves are random, as $Y_{n +t}$ and $c_{n + t}$ are both random. However, our goal remains the same: ensure $\FDR(\rejset_t) \leq \alpha$ for each $t \in \naturals$ where the expectation is now also taken over the randomness of whether a hypothesis is null or not.
As argued in~\cite{jin_model-free_selective_2023}, this type of selection occurs widely in practice, e.g., screening for high performing job candidates based on interview performance, picking patients with attributes that are responsive to treatment, detecting outliers, etc.
In this setting with randomized null hypotheses, we require our p-values and e-values to satisfy the following conditions instead for each $t \in \naturals$:
\begin{gather}
    \prob{P_t \leq \alpha, Y_{n +t} \leq c_{n + t}} \leq \alpha \text{ for all }\alpha \in [0, 1], \label{eq:joint-superuniform}\\
    \expect[E_t \cdot \ind{Y_{n + t}\leq c_{n + t}}] \leq 1. \label{eq:joint-evalue}
\end{gather}

In addition, $\mathbf{Q}$ results from a covariate shift on $\mathbf{P}$. This means that $\mathbf{P}(Y \mid X = x) = \mathbf{Q}(Y\mid X = x)$ for all $x \in \Xcal$. Further, the Radon-Nikodym derivative (w.r.t.\ to an arbitrary common base measure) satisfies $(d\mathbf{Q} / d\mathbf{P})(x, y) = w(x)$ for all $x \in \Xcal$, where $w$ is a likelihood ratio dependent only on $x \in \Xcal$. We assume we have access to $w$ (e.g., we can esimate it from other data accurately). In addition, define  a \emph{monotone} score function $V: \Xcal \times \Ycal \mapsto \reals$ as a function satisfying $V(x, y) \leq V(x, y')$ for all $x \in \Xcal$ and $y, y' \in \Ycal$ where $y \leq y'$.

\subsection{FDR control through online multiple testing}
\citet{jin_model-free_selective_2023} construct the following p-value using any monotone score function $V$:
\begin{gather}
    V_i \coloneqq V(X_i, Y_i), \qquad \widehat{V}_{n + t} \coloneqq V(X_{n + t}, c_{n + t}) ,\\
P_t \coloneqq \frac{\sum_{i = 1}^n w(X_i) \mathbf{1}\{V_i < \widehat{V}_{n + t}\} + w(X_{n + t})}{\sum_{i = 1}^n w(X_i) + w(X_{n + t})},
\label{eq:weighted-pvalue}
\end{gather}
For simplicity, we assume that neither $(V_i)_{i \in [n]}$ nor $(\widehat{V}_{n + t})_{t \in \naturals}$ have point masses in their distributions in this paper, and this assumption can be relaxed through simple modifications  to the p-value formulations \citep[eqs. 3 \& 6]{jin_model-free_selective_2023}.
\begin{fact}[Lemma 2.2 \citep{jin_model-free_selective_2023}]
    For each $t \in \naturals$, $P_t$ defined in \eqref{eq:weighted-pvalue} is a p-value \eqref{eq:joint-superuniform}.
    \label{fact:wcs-pvalue}
\end{fact}
The dependence structure among $(P_t)$ is quite complicated, and does not satisfy usual independence or positive dependence notions that are amenable to multiple testing without correction \cite[Proposition 2.4]{jin_model-free_selective_2023}.
Thus, one must apply \rLOND\ (or \UrLOND) derive FDR control.
\begin{proposition}
    Let $(\rejset_t^{\rLOND})$ and $(\rejset_t^{\UrLOND})$ be the sequences of rejection sets that arise from applying \rLOND\ or \UrLOND, respectively, to $(P_t)$ as defined in \eqref{eq:weighted-pvalue}. Then, $\FDR(\rejset_t^{\rLOND}) \leq \alpha$ and $\FDR(\rejset_t^{\UrLOND}) \leq \alpha$ for each $t \in \naturals$.
    \label{prop:wcs-p-fdr}
\end{proposition}

We defer the proof of this result to \Cref{sec:wcs-p-fdr-proof}. \citet{jin_model-free_selective_2023} show that the more powerful way to utilize $P_t$ is to view them as e-values, and we show that a similar phenomenon is also possible for online WCS\@. First, define the following leave-one-out conformal p-values $P^{(t), -}_{j},P^{(t), +}_{j}$ for each $t \in \naturals$ and $j \in [t - 1]$:
\begin{align}
P_j^{(t),-} &\coloneqq \frac{\sum_{i = 1}^n w(X_i) \mathbf{1}\{V_i < \widehat{V}_{n + j}\}}{\sum_{i = 1}^n w(X_i) + w(X_{n + t})},\\
    P_j^{(t),+} &\coloneqq \frac{\sum_{i = 1}^n w(X_i) \mathbf{1}\{V_i < \widehat{V}_{n + j}\} + w(X_{n + t})}{\sum_{i = 1}^n w(X_i) + w(X_{n + t})}.
\end{align}

Let $\widehat{\rejset}^{\LOND(t), -}_{t - 1}$ and $\widehat{\rejset}^{\LOND(t), +}_{t - 1}$ be the discovery set obtained from applying \LOND\ to
$(P_j^{(t), -})_{j \in [t - 1]}$ and
$(P_j^{(t), +})_{j \in [t - 1]}$, respectively.
Define the test levels for the next hypothesis as
\begin{align}
    \widehat{\alpha}_t^{\LOND, -} \coloneqq \alpha \gamma_t \cdot (|\widehat{\rejset}_{t - 1}^{(t), -}| + 1),\qquad
    \widehat{\alpha}_t^{\LOND, +} \coloneqq \alpha \gamma_t \cdot (|\widehat{\rejset}_{t - 1}^{(t), +}| + 1).
\end{align} We can now define the following e-value:
\begin{align}
    E_t^\LOND \coloneqq \mathbf{1}\{P_t \leq \widehat{\alpha}_t^{\LOND, +}\} / \widehat{\alpha_t}^{\LOND, -}.
    \label{eq:weighted-evalue}
\end{align}

\begin{proposition}\label{prop:wcs-evalue}
    For each $t \in \naturals$, $E_t^\LOND$ is an e-value~\eqref{eq:joint-evalue}.
\end{proposition}
We defer the proof of this result to~\Cref{sec:wcs-evalue-proof}. We can derive the FDR control of \ELOND\ or \ULOND\ applied to $(E_t^{\LOND})$\@.
\begin{theorem}\label{thm:wcs-evalue-fdr}
    Using e-values $(E_t)$ satisfying~\eqref{eq:joint-evalue}, $\FDR(\rejset_t^\ELOND) \leq \alpha$ and $\FDR(\rejset_t^\ULOND) \leq \alpha$ for each $t \in \naturals$. \end{theorem}
We defer the proof of this result to~\Cref{sec:wcs-evalue-fdr-proof}. Now, we apply our online WCS techniques to some real data settings in~\cite{jin_model-free_selective_2023}, and use their code to calculate the weighted p-values in~\eqref{eq:weighted-pvalue} for each setup.
 \begin{figure}[h]\label{fig:wcs-results}
    \begin{subfigure}{0.5\textwidth}
    \includegraphics[width=\textwidth]{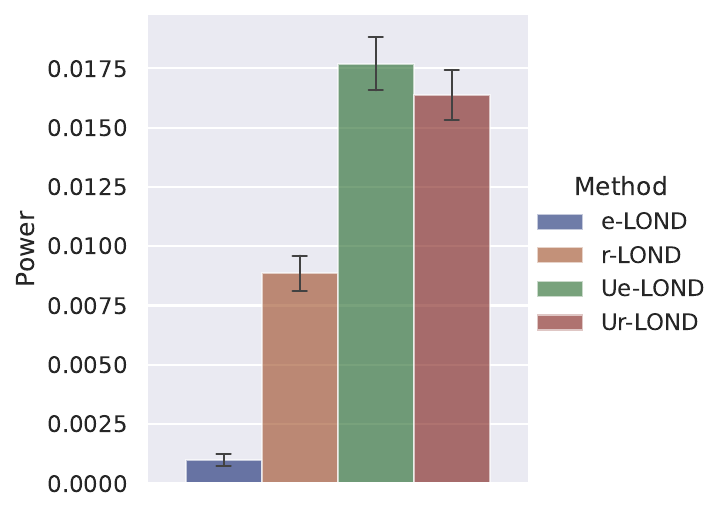}
    \end{subfigure}
    \begin{subfigure}{0.5\textwidth}
    \includegraphics[width=\textwidth]{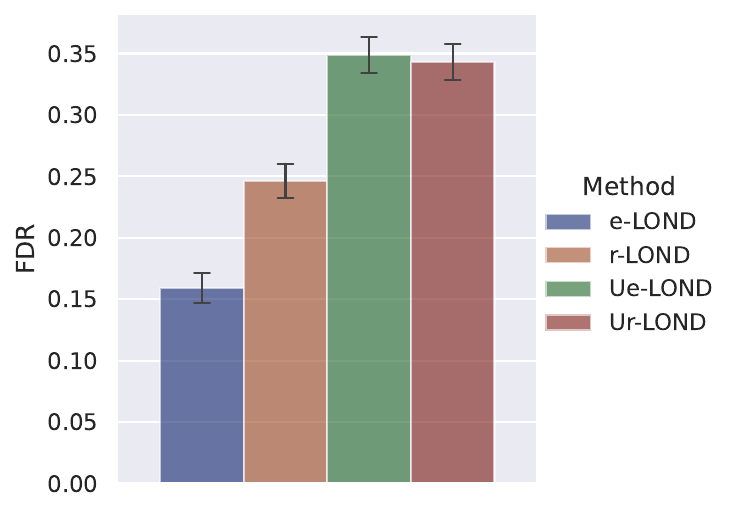}
    \end{subfigure}
    \caption{Average power and empirical FDR for methods applied at $\alpha=0.5$ for the drug property prediction task, with the error bars marking one standard error from Monte Carlo estimation. We can see that \ULOND\ has the largest power. E-LOND has the smallest power, due to the way $(E_t^\LOND)$ are formulated to almost always be below the test threshold of \ELOND\ itself. However, by allowing randomization in \ULOND, we see that this issue is fixed and \ULOND\ exceeds the power of both \UrLOND\ and \rLOND.}
\end{figure}

\subsection{Drug property prediction}
We tackle the task of predicting drug properties that uses the HIV screening dataset in the DeepPurpose library \citep{huang_deeppurpose_deep_2021} --- the goal is to select a subset of drug candidates that bind to a target protein for HIV. The covariate $X$ is the chemical structure of the drug, that is encoded into the form of a vector $\reals^d$, and $Y \in \{0, 1\}$ is binary label of whether it does not or does bind. In constructing the calibration set, experimenters might pick drugs that seem more likely to bind to analyze (and label) and induce a covariate shift as a result of selection bias. Thus, we construct a setup that emulates this issue.
40\% of the data set is placed in $\Dcal_\train$ and used to train a neural network classifier $\widehat{\mu}: \Xcal \mapsto \reals$ that predicts the probability of binding. 60\% more of the dataset is used to construct $\Dcal_{\calib}$ by selecting each point $(X_i , Y_i)$ to be in $\Dcal_\calib$ with probability $p(X_i)$ where $p(x) = \sigma(\widehat{\mu}(x) - \bar{\mu}) \wedge 0.8$.
Of the points that are neither in $\Dcal_{\train}$ nor $\Dcal_{\calib}$, we sample 5\% randomly to constitute $\Dcal_{\test}$ due to computational constraints. Consequently, there is a covariate shift between the calibration and the test set, and the resulting likelihood ratio satisfies $w(x)\propto 1 / p(x)$. The null hypothesis that we wish to test is as follows:
 \begin{align}
     H_0^t: Y_{n + t}  = 0.
 \end{align}
Controlling the FDR results in selecting a subset of drugs where only a small proportion do not bind to the protein in expectation. We average our results over 600 trials. We see that the power of \ULOND\ in \Cref{fig:wcs-results} is the largest. On the other hand, the power of \ELOND\ is the smallest. This is because $E_t^\LOND$ is either 0 or $1 / \widehat{\alpha}_t^{\LOND, +}$, and $\alpha_t^\ELOND < \widehat{\alpha}_t^{\LOND, +}$ holds often, as $\widehat{\alpha}_t^{\LOND, +}$ is a conservative estimate of $\alpha_t^\ELOND$. The randomization from \ULOND\ alleviates this problem, hence it attaining the largest power. All methods also practically control FDR at the desired level of $\alpha=0.5$.
 \section{Related work}
\label{sec:RelatedWork}

This work lies at the intersection of e-values and online multiple testing. We outline the most relevant research in each of these areas to this work.

\paragraph{Online multiple testing} Online multiple testing was first posed by \cite{foster2008alpha} when they were studying computationally cheap methods for performing streamed variable selection in high dimensional things and proved mFDR control for alpha-investing. The methods were subsequently improved in several follow up works to be more powerful and also guarantee control of the FDR \cite{aharoni2014generalized,javanmard2018online,ramdas2017online}. \citet{ramdas2018saffron} and \citet{tian2019addis} developed adaptive online multiple testing procedures based on Storey's method \citep{storey_false_discovery_2002} for offline FDR control. With the exception of \citet{javanmard2018online}, all these works all focus on online FDR or mFDR control under the assumption that p-values are independent or are p-values when conditioned on the information observed so far (e.g., previous p-values, rejection decisions, etc.), i.e., conditional superuniformity. As mentioned before, more recent work of \citet{zrnic_asynchronous_2021} considers explicitly modeling dependence relationships through conflict sets to derive algorithms that still control mFDR and FDR even when independence or conditional superuniformity is not satisfied. Another line of work considers the situation when the rejection decision of a hypothesis does not have to be made immediately, but only need to be made by a later time, such as at the end of a batch of hypotheses being jointly experimented on \citep{zrnic2020power} or at individual future deadlines \citep{fisher_online_control_2022}. This is the first work that directly targets the arbitrary dependence case. \citet{robertson2019onlinefdr} provide a R package implementing many of the aforementioned methods for online control of the FDR, in addition to other online multiple testing methods. Online multiple testing methods (including \LOND) have already been applied in a variety of medicinal and biological applications \citep{robertson2018online,robertson_online_error_2022,liou_global_fdr_2023}. 

\paragraph{E-values} E-values have been applied in many offline multiple testing settings such as FDR control \citep{wang_false_discovery_2020,ignatiadis_e-values_unnormalized_2022} and closed testing \citep{vovk_confidence_discoveries_2023}. In particular, the e-BH procedure introduced by \citet{wang_false_discovery_2020} has been used as a subroutine in other multiple testing procedures with FDR control such as in the bandit setting \citep{xu_unified_framework_2021}, for the purpose of derandomizing knockoffs \citep{ren_derandomized_knockoffs_2022} or achieving optimality under a Bayesian linear model alternative \citep{ahn_nearoptimal_multiple_2022}. \citet{xu_postselection_inference_2022} present selective inference procedure with FCR control for e-CIs. Further, \citet{jin_model-free_selective_2023} showed that the weighted conformal selection procedure in their paper can also be viewed as an application of e-BH to e-values. This work is novel in bringing all these insights concerning e-values that have been used in offline multiple testing to the online setting.  \section{Omitted proofs}

Here, we include the full proofs of the results contained in \Cref{sec:ELOND}, \Cref{sec:FCR}, and \Cref{sec:conformal}.
\subsection{Proof of \Cref{thm:ELONDArbDep}}\label{sec:ArbDepProof}

For brevity, we will write \(\alpha_t^{\ELOND}\) as \(\alpha_t\) in the proofs in this section.
\begin{align}
    \FDR(\rejset_t) &= \expect\left[\sum\limits_{i \in \hypset_0 \cap [t]}\frac{\ind{E_i \geq \alpha_i^{-1}}}{|\rejset_{t}| \vee 1}\right]
    =\sum\limits_{i \in \hypset_0 \cap [t]}\expect\left[\frac{\ind{E_i \geq \alpha_i^{-1}}}{|\rejset_{t}| \vee 1} \times \ind{E_i \geq \alpha_i^{-1}}\right]\\
    &\labelrel{\leq}{rel:rej-ub}\sum\limits_{i \in \hypset_0 \cap [t]}\expect\left[\frac{\alpha_i E_i}{|\rejset_{t}| \vee 1} \times \ind{|\rejset_t| \geq |\rejset_{i - 1}| + 1}\right]\\
    &\labelrel{\leq}{rel:elond-def}\sum\limits_{i \in \hypset_0 \cap [t]}\expect\left[\frac{\alpha \gamma_i(|\rejset_{i - 1}| + 1) E_i}{|\rejset_{i - 1}| + 1} \times \ind{|\rejset_t| \geq |\rejset_{i - 1}| + 1}\right]\\
    &\labelrel{\leq}{rel:drop-ind}\sum\limits_{i \in \hypset_0 \cap [t]}\expect\left[\frac{\alpha \gamma_i(|\rejset_{i - 1}| + 1) E_i}{|\rejset_{i - 1}| + 1} \right] = \alpha \sum\limits_{i \in \hypset_0 \cap [t]}\gamma_i\expect\left[E_i\right] \tag*{} \leq \alpha.
\end{align}

Inequality \eqref{rel:rej-ub} is a result of \eqref{eqn:DeterministicInequality} and $|\rejset_t| \geq |\rejset_{i - 1}| + \ind{E_i \geq \alpha_i^{-1}}$ by construction of $\rejset_t$. Inequality \eqref{rel:rej-ub} is a result of the indicator in the expectation (i.e., making discovery at $H_i$ will make $\rejset_t$ larger than $\rejset_{i - 1}$).
Inequality \eqref{rel:drop-ind} comes from dropping the indicator term.
The last inequality is due to \(\expect[E_t] \leq 1\) for all \(t \in \hypset_0\) by definition of e-values \eqref{eqn:EMarginalExpect}, and because \((\gamma_t)\) sum up to 1. Thus, we achieve an upper bound of \(\alpha\) on the final line and have shown our desired result on FDR.

To show \ELOND\ strictly dominates \rLOND, it is sufficient show that \(\alpha_t^{\ELOND} \geq \alpha_t^{\rLOND}\) for all \(t \in \naturals\), and there exists a sequence of e-values $(E_t)$ such that there exists $t \in \naturals$ such that $\alpha_t^{\ELOND} > \alpha_t^{\rLOND}$. For any $t \in \naturals$,

\begin{align}
\beta_t(|\rejset_{t - 1}| + 1) = \int\limits_0^{|\rejset_{t - 1}| + 1} x\ d\nu(x) \leq |\rejset_{t - 1}| + 1,
\end{align} where the first equality is by definition of reshaping function, and the inequality is because $x \leq |\rejset_{t - 1}| + 1$ in the integrand, and $\nu$ is a probability measure that is nonnegative and integrates to 1. Thus, \(\alpha_t^{\ELOND} \geq \alpha_t^{\rLOND}\) for all \(t \in \naturals\).

Next, note for $\beta_2$, either it satisfes (1) $ \beta_2(2) = 2$ and $\beta_2(1) = 0$ or (2) $\beta_2(2) < 2$ --- this follows from the definition of reshaping function, and case (1) correpsonds to putting all probability mass in $\nu$ on 2.

If $\beta_2$ satisfies case (1), then we set $E_1 = 1 / (\alpha \gamma_1) + 1$. This results in $\alpha_2^{\ELOND} = \alpha \gamma_2 > 0 = \alpha_2^{\rLOND}$. Otherwise, we set $E_1 = 1 / (\alpha \gamma_1)$, which leads to a rejection by \ELOND, and note that $\alpha_2^{\rLOND} \leq \alpha\gamma_2\beta_2(2) < 2\alpha\gamma_2 = \alpha_2^{\ELOND}$. Thus, we have shown that \ELOND\ strictly dominates \rLOND\ applied to $(1 / E_t)$ and conclude our proof.\hfill$\qed$

\subsection{Proof of \Cref{thm:UrLOND}}\label{sec:urlond-proof}
For simplicity, denote $\alpha_t^\UrLOND, \rejset_t^\UrLOND$ as $\alpha_t, \rejset_t$. Similar to the proof of FDR control for \rLOND\ in \cite{zrnic_asynchronous_2021}, we first show the following inequality for any $i \in [t]$:\begin{align}
        &\expect\left[\frac{\ind{P_i \leq \alpha_i^{\UrLOND}}}{|\rejset_t| \vee 1}\right] \labelrel{=}{rel:ind-imp} \expect\left[\frac{\ind{P_i \leq \alpha_i^{\UrLOND}}}{|\rejset_t| \vee 1}\ind{|\rejset_{t - 1}| \geq |\rejset_{i - 1}| + 1}\right]\\
        &\labelrel{=}{rel:expand} \expect\left[\frac{\ind{P_i \leq \alpha \gamma_i\beta_i((|\rejset_{i - 1}| + 1) / U_i)}}{|\rejset_{t}| \vee 1}\ind{|\rejset_t| \vee 1 \geq |\rejset_{i - 1}| + 1}\right]\\
        &\labelrel{\leq}{rel:denom-ub} \expect\left[\frac{\ind{P_i \leq \alpha \gamma_i\beta_i((|\rejset_{i - 1}| + 1) / U_i)}}{|\rejset_{i - 1}| + 1}\ind{|\rejset_t| \vee 1 \geq |\rejset_{i - 1}| + 1}\right]\\
        &\labelrel{\leq}{rel:denom-drop-ind} \expect\left[\frac{\ind{P_i \leq \alpha \gamma_i\beta_i((|\rejset_{i - 1}| + 1) / U_i)}}{|\rejset_{i - 1}| + 1}\right]
        \labelrel{\leq}{rel:su-apply} \alpha \gamma_t.\label{eq:ind-bound}
        \label{eq:fdr-ind-bound}
    \end{align} Equality \eqref{rel:ind-imp} is because $\{P_i \leq \alpha_i^{\UrLOND}\} \Rightarrow \{|\rejset_{i - 1}| + 1 \leq |\rejset_t| \vee 1\}$ as a result of a discovery being made at the $i$th hypothesis.
    Equality \eqref{rel:expand} is by expanding the definition of $\alpha_i^{\UrLOND}$.
    Inequality \eqref{rel:denom-ub} is the indicator $\ind{|\rejset_{t}|\vee 1 \geq |\rejset_{i - 1}| + 1}$ being 1 iff the event it is indicating is true.
    Inequality \eqref{rel:denom-drop-ind} is simply by droppign the indicator.
    Inequality \eqref{rel:su-apply} is by \Cref{fact:rand-su-lemma}. Thus, we can derive the following bound on the FDR by \eqref{eq:ind-bound}:
    \begin{align}
        \FDR(\rejset_t) &= \sum\limits_{i \in \hypset_0 \cap [t]} \expect\left[\frac{\ind{P_i \leq \alpha_i^{\UrLOND}}}{|\rejset_t| \vee 1}\right]\\
        &\leq \alpha \sum\limits_{[t]} \gamma_t\leq \alpha,
    \end{align} which achieves our desired FDR control.

    The strict dominance in expectation follows from the fact that $\alpha_t^{\UrLOND} > \alpha_t^{\rLOND}$ with nonzero probability whenever $|\rejset_{t - 1}| < t - 1$ because $U_t^{-1}$ is a positive number that is at least 1, and $(|\rejset_{t - 1}| + 1)U_t^{-1} \geq |\rejset_{t - 1}| + 2$ (which implies $\beta_t^\BY((|\rejset_{t - 1}| + 1)U_t^{-1}) > \beta_t^\BY(|\rejset_{t - 1}| + 1)$) with nonzero probability. Thus, we have shown strict dominance in expectation and all results in the theorem.\hfill$\qed$

\subsection{Proof of \Cref{thm:FCR}}\label{sec:FCRProof}
Denote $\alpha_t^{\ELOND}$ as $\alpha_t$ in this section.
We make the following derivation for the FCR:
\begin{align}
    \FCR(\Scal_t) &= \expect\left[\sum\limits_{i \in \Scal_t}\frac{\ind{\theta_i \not\in C_i(X_i, \alpha_i)}}{|\Scal_{t}| \vee 1}\right]
    \labelrel{=}{rel:ci-ind}\expect\left[\sum\limits_{i \in \Scal_t}\frac{\ind{E_{\theta_i} \geq \alpha_i^{-1}}}{|\Scal_{t}| \vee 1}\right]\\
&\labelrel{\leq}{rel:elond-ci-def}\expect\left[\sum\limits_{i \in \Scal_t}\frac{\alpha \gamma_i(|\Scal_{i - 1}| + 1) E_i}{|\Scal_t| \vee 1} \right]
    \labelrel{=}{rel:ind-move}\sum\limits_{i \in [t]}\expect\left[\frac{\alpha \gamma_i(|\Scal_{i - 1}| + 1) E_i \ind{i \in \Scal_t}}{|\Scal_t| \vee 1} \right]\\
    &\labelrel{\leq}{rel:sel-grow}\sum\limits_{i \in [t]}\expect\left[\frac{\alpha \gamma_i(|\Scal_{i - 1}| + 1)E_{\theta_i}\ind{|\Scal_t| \geq |\Scal_{i - 1}| + 1}}{|\Scal_t| \vee 1} \right]\\
    &\labelrel{\leq}{rel:shrink-denom}\sum\limits_{i \in [t]}\expect\left[\frac{\alpha \gamma_i(|\Scal_{i - 1}| + 1)E_{\theta_i}\ind{|\Scal_t| \vee 1 \geq |\Scal_{i - 1}| + 1}}{|\Scal_{i - 1}| + 1} \right]\\
    &\labelrel{\leq}{rel:ci-drop-ind}\sum\limits_{i \in [t]}\expect\left[\frac{\alpha \gamma_i(|\Scal_{i - 1}| + 1) E_{\theta_i}}{|\Scal_{i - 1}| + 1} \right] = \alpha \sum\limits_{i \in [t]}\gamma_i\expect\left[E_{\theta_i}\right] \leq \alpha.
\end{align}
Equality~\eqref{rel:ci-ind} is by the definition of an e-CI in~\eqref{eq:eci-def}.
Inequality~\eqref{rel:elond-ci-def} is by the definition of an e-LOND.
Equality~\eqref{rel:ind-move} is simply arithmetic with the indicator of whether $i$ is in $\Scal_t$.
Inequality~\eqref{rel:sel-grow} is because $i \in \Scal_t$ implies that $\Scal_{t}$ gained a selected parameter, namely the $i$th parameter, over $\Scal_{i - 1}$.
Inequality~\eqref{rel:shrink-denom} is because $i \in \Scal_t$ implies that $\Scal_{t}$ gained a selected parameter, namely the $i$th parameter, over $\Scal_{i - 1}$.
Inequality~\eqref{rel:ci-drop-ind} follows from dropping the indicator, and the last inequality is again due to \(\expect[E_{\theta_i}] \leq 1\) for each $i \in \naturals$ by definition of e-values \eqref{eqn:EMarginalExpect}, and because \((\gamma_t)\) sum up to 1. Thus, we achieve our desired result of FCR control of \(\alpha\).

\ULOND-CI can be shown to have FCR control by following the above argument, except we can replace $E_{\theta_i}$ with $S_{\alpha_t^\ELOND}(E_{\theta_i})$. Thus, we have shown our desired levels of FCR control. \hfill$\qed$

\subsection{Proof of \Cref{prop:wcs-p-fdr}}
\label{sec:wcs-p-fdr-proof}
Let $\widetilde{P}_t \coloneqq P_t \vee \ind{Y_t > c_t}$. Note that for each $t \in \naturals$, $\widetilde{P}_t$ satisfies the following two properties.
\begin{gather}
    \{\widetilde{P}_t\leq s\} \Leftrightarrow \{P_t \leq s, Y_t \leq c_t\}\text{ and }\prob{\widetilde{P}_t\leq s} = \prob{P_t\leq s, Y_t \leq c_t} \leq s \text{ for all }s \in [0, 1). \label{eq:weighted-p-ub}
\end{gather} This is by definition of $\widetilde{P}_t$ and by the superuniform constraint on $P_t$ in \eqref{eq:joint-superuniform}. Further, we can see that
\begin{align}
    \{P_t \leq s, Y_t \leq c_t\} \Rightarrow \{\widetilde{P}_t\leq s\} \text{ when }s = 1, \label{eq:weighted-p-except}
\end{align} by definition of $\widetilde{P}_t$ as well.

We also observe the following implication holds:
\begin{align}
    \{\widetilde{P}_i \leq \alpha_i\} \Rightarrow \{\rejset_t \supset \rejset_{i - 1} \} \Rightarrow \{|\rejset_t| \vee 1 \geq \rejset_{i - 1} + 1\},\label{eq:rej-grow}
\end{align} for all $t \geq i$ simply because a discovery set grows when a a new discovery is made.

Let $(\alpha_t), (\rejset_t)$ be either $(\alpha_t^\rLOND), (\rejset_t^\rLOND)$ or $(\alpha_t^\UrLOND), (\rejset_t^\UrLOND)$. We can make the following derivation of the FDR:
\begin{align}
    \FDR(\rejset_t) &= \sum\limits_{i \in [t]}\expect\left[\frac{\ind{P_i \leq \alpha_t, i \in \hypset_0}}{|\rejset_t| \vee 1}\right] = \sum\limits_{i \in [t]}\expect\left[\frac{\ind{P_i \leq \alpha_t, Y_i \leq c_i}}{|\rejset_t| \vee 1}\right]
                    \labelrel{\leq}{rel:ind-swap} \sum\limits_{i \in [t]}\expect\left[\frac{\ind{\widetilde{P}_i \leq \alpha_i}}{|\rejset_t| \vee 1}\right]\\
                    &\labelrel{=}{rel:rej-grow-ind} \sum\limits_{i \in [t]}\expect\left[\frac{\ind{\widetilde{P}_i \leq \alpha_i}}{|\rejset_{i - 1}| + 1}\right]
                    \labelrel{\leq}{rel:denom-swap} \sum\limits_{i \in [t]}\expect\left[\frac{\ind{\widetilde{P}_i \leq \alpha \gamma_i \cdot \beta_i((|\rejset_{i - 1}| + 1) / U_i)}}{|\rejset_{i - 1}| + 1}\right]
                    \labelrel{\leq}{rel:gamma-sum} \sum\limits_{i \in [t]}\alpha \gamma_i \leq \alpha.
\end{align}
Inequality \eqref{rel:ind-swap} is by a combination of \eqref{eq:weighted-p-ub} and \eqref{eq:weighted-p-except}.
Inequality \eqref{rel:rej-grow-ind} is because of \eqref{eq:rej-grow}.
Inequality \eqref{rel:denom-swap} is by the definition of either choice of $(\alpha_t)$ ($U_i = 1$ if \rLOND, and $U_i$ is an independent uniform random variable over $[0, 1]$ if \UrLOND) and the fact that $|\rejset_t| \vee 1 \leq |\rejset_{t - 1}| + 1$ by definition of discovery sets.
Inequality \eqref{rel:gamma-sum} is by \Cref{fact:rand-su-lemma}, since $U_i$ is superuniform and independent of all $\widetilde{P}_i$.
The last inequality is due to $\sum_{i \in [t]}\gamma_i \leq 1$. Thus, we have shown our desired FDR control.\hfill$\qed$

\subsection{Proof of \Cref{prop:wcs-evalue}}
\label{sec:wcs-evalue-proof}

We follow a similar proof structure to the proof of Theorem 3.1 in \citet{jin_model-free_selective_2023}.

First, we define the following oracle p-values (that cannot be computed from the observable data) to assist with our proof:
\begin{align}
    \bar{P}_t  &\coloneqq  \frac{\sum_{i = 1}^n w(X_i) \mathbf{1}\{V_i < V_{n + t}\} + w(X_{n + t})}{\sum_{i = 1}^n w(X_i) + w(X_{n + t})}.\\
    \bar{P}_j^{(t)}  &\coloneqq  \frac{\sum_{i = 1}^n w(X_i) \mathbf{1}\{V_i < \widehat{V}_{n + j}\} + w(X_{n + t})\mathbf{1}\{V_{n + t} < \widehat{V}_{n + j}\}}{\sum_{i = 1}^n w(X_i) + w(X_{n + t})}.
\end{align} These essentially replace $\widehat{V}_{n + t}$ with $V_{n + t}$ when compared to their empirical counterparts $P_t$ and $P_j^{(t)}$, respectively.
The first thing we note is the following relationship between the oracle nonconformity score and the empirical nonconformity score at $n + t$:
\begin{align}
    t \in \hypset_0 \Leftrightarrow Y_{n + t} \leq c_{n + t} \Rightarrow V_{n + t} \leq \widehat{V}_{n + t} \Rightarrow \bar{P}_t \leq \widehat{P}_t,
\end{align} since $V$ is a monotone score function. Further, the oracle p-values $(\bar{P}_j^{(t)})_{j \in [t - 1]}$ are bounded by their empirical counterparts, i.e.,
\begin{align}
   \widehat{P}_j^{(t), -} \leq \bar{P}_j^{(t)} \leq \widehat{P}_j^{(t), +} \text{ for all }t \in \naturals\text{ and }j \in [t - 1].\label{eq:oracle-lb-emp}
\end{align}

Define $\bar{\rejset}_{t - 1}$ to be the discovery set that results from applying $\LOND$ to $(\bar{P}_1^{(t)}, \dots, \bar{P}_{t - 1}^{(t)})$, and define
\begin{align}
    \bar{\alpha}_t^\LOND \coloneqq \alpha \gamma_t \cdot (|\bar{\rejset}_{t - 1}| + 1), \qquad \bar{E}_t^{\LOND} \coloneqq \ind{\bar{P}_t \leq \bar{\alpha}_t^\LOND} / \bar{\alpha}_t^\LOND
\end{align} to be the test level for the next hypothesis and an all-or-nothing e-value testing at that level, respectively. By \eqref{eq:oracle-lb-emp}, we can derive that
\begin{align}
    |\widehat{\rejset}_{t - 1}^+| \leq |\bar{\rejset}_{t - 1}| \leq |\widehat{\rejset}_{t - 1}^-|, \text{ and } \widehat{\alpha}_t^{\LOND, +} \leq \bar{\alpha}_t^\LOND \leq \widehat{\alpha}_t^{\LOND, -}.\label{eq:rej-ineq}
\end{align}
This gives us the following inequality:
\begin{align}
    \ind{t \in \hypset_0} \cdot E_t^\LOND = \frac{\ind{t \in \hypset_0} \cdot \ind{\widehat{P}_t \leq \widehat{\alpha}_t^{\LOND, +}}}{\widehat{\alpha}_{t}^{\LOND, -}}
    \leq \frac{\ind{t \in \hypset_0}\cdot\ind{ \bar{P}_t \leq \bar{\alpha}_t^\LOND}}{\bar{\alpha}_t^\LOND}\leq \bar{E}_t^\LOND.\label{eq:oracle-e-ub}
\end{align}

Now we need to show that $\bar{E}_t^\LOND$ is an e-value as defined in \eqref{eq:joint-evalue}. Define $Z_i \coloneqq (X_i, Y_i)$ for each $i \in \naturals$. Let $Z\coloneqq [Z_1, \dots, Z_n, Z_{n + t}]$ denote the unordered set of $\{Z_1, \dots, Z_n, Z_{n + t}\}$, and $z = [z_1, \dots, z_n, z_{n + t}]$ be the unordered set of their realized values. Define $\xi_{z, t}$ as the event such that $Z = z$. Let $I_t \in [n] \cup \{n + t\}$ be the index such that $Z_{n + t} = z_{I_t}$. Now, we note the following important facts
\begin{gather}
    \bar{P}_t \text{ is measurable w.r.t.\ }Z\text{ and }I_t.\\
    (\bar{P}_j^{(t)})_{j \in [t - 1]}, \bar{\rejset}_{t - 1}, \bar{\alpha}_t^\LOND \text{ are measurable w.r.t.\ }Z\text{ and }\{Z_{n + i}\}_{i \neq t}.
\end{gather}
In addition, we have that
\begin{align}
    \{Z_{n + i}\}_{i \neq t } \indep I_t \mid \xi_{t, z}.
\end{align} This is a result of $\{Z_{n + i}\}_{i \neq t } \indep \{Z_i\}_{i \in [n] \cup \{n + t\}}$ since each data point is assumed to be independent. As a result, we can conclude that
\begin{align}
    \bar{P}_t \indep \bar{\alpha}_t^\LOND \mid \xi_{z, t}.\label{eq:cond-ind}
\end{align}
Let $F_{z, t} \coloneqq \prob{\bar{P}_t \leq \bar{\alpha}_t^\LOND \mid \xi_{z, t}}$ be the conditional c.d.f.\ of $\bar{P}_t$.

Now, we define a randomized oracle conformal p-value:
\begin{align}
    P^*_t\coloneqq  \frac{\sum_{i = 1}^n w(X_i) \mathbf{1}\{V_i < V_{n + t}\} + U^*_t (w(X_{n + t}) + \ind{V_i = V_{n + t}})}{\sum_{i = 1}^n w(X_i) + w(X_{n + t})}.
\end{align} where $U_t^*$ is an independent uniform random variable on $[0, 1]$.

We know cite the following fact from \citet{hu_two-sample_conditional_2023} that arises due to weighted exchangeability of $(Z_1, \dots, Z_n, Z_{n + t})$:
\begin{fact}[Lemmas 2 and 3 of \citet{hu_two-sample_conditional_2023}]
    $P^*_t \mid \xi_{t, j}$ is uniformly distributed over $[0, 1]$.
\end{fact}
Since $P^*_t \leq \bar{P}_t$ determinstically, we have that
\begin{align}
    F_{z, t}(s) \leq \prob{P^*_t \leq s \mid \xi_{z, t}} \leq s\text{ for all }s \in [0, 1].\label{eq:superuniform-f}
\end{align}

Relating this back to our e-value, we, get that
\begin{align}
    \expect[\bar{E}_t^\LOND \mid \xi_{z, t}] = F_{z, t}(\bar{\alpha}^\LOND_t) / \bar{\alpha}^\LOND_t \leq 1 \label{eq:cond-ub}
\end{align} by \eqref{eq:superuniform-f} and \eqref{eq:cond-ind}. $\expect[\bar{E}_t^\LOND] \leq 1$ follows by the tower property of conditional expectation applied to \eqref{eq:cond-ub}. Hence, our desired result that $E_t^\LOND$ is an e-value follows from \eqref{eq:oracle-e-ub}.

\subsection{Proof of \Cref{thm:wcs-evalue-fdr}}
\label{sec:wcs-evalue-fdr-proof}
Let $\alpha_t, \rejset_t$ be short for $\alpha_t^\ELOND, \rejset^{\ELOND}_t$. We can make the following derivation:

\begin{align}
    \FDR(\rejset_t) &= \sum\limits_{i \in [t]}\expect\left[\frac{\ind{E_i \geq \alpha_i, i \in \hypset_0}}{|\rejset_t| \vee 1}\right]
= \sum\limits_{i \in [t]}\expect\left[\frac{\ind{E_i \geq \alpha_i}\cdot \ind{i \in \hypset_0}}{|\rejset_t| \vee 1}\right]\\
                    &\labelrel{=}{rel:grow-rejset} \sum\limits_{i \in [t]}\expect\left[\frac{\ind{E_i \geq \alpha_i}\cdot \ind{i \in \hypset_0}}{|\rejset_t| \vee 1} \cdot \ind{|\rejset_t| \geq |\rejset_{i - 1}| + 1}\right]\\
                    &\labelrel{\leq}{rel:det-ub} \sum\limits_{i \in [t]}\expect\left[\frac{\alpha_i E_i\cdot \ind{i \in \hypset_0}}{|\rejset_t| \vee 1} \cdot \ind{|\rejset_t| \vee 1 \geq |\rejset_{i - 1}| + 1}\right]\\
                    &\labelrel{\leq}{rel:drop-rejset-ind} \sum\limits_{i \in [t]}\expect\left[\frac{\alpha_i E_i\cdot \ind{i \in \hypset_0}}{|\rejset_{i - 1}| + 1} \right]
                    \labelrel{=}{rel:expand-def}\sum\limits_{i \in [t]}\expect\left[\frac{\alpha \gamma_i(|\rejset_{i-1}| + 1) E_i\cdot \ind{i \in \hypset_0}}{|\rejset_{i - 1}| + 1} \right]\\
                    &=\sum\limits_{i \in [t]}\alpha \gamma_i\expect\left[ E_i\cdot \ind{i \in \hypset_0} \right]
\leq \sum\limits_{i \in [t]}\alpha \gamma_i \leq \alpha.
\end{align}
Inequality \eqref{rel:grow-rejset} is because $E_i\geq \alpha_i$ implies a discovery is made at the $i$th hypothesis. Inequality \eqref{rel:det-ub} is because $E_i, \alpha_i$ are nonnegative. Inequality \eqref{rel:drop-rejset-ind} is a result of dropping the indicator for $|\rejset_t| \vee 1 \geq |\rejset_{i - 1}| + 1$ and lower bounding the denominator. Equality \eqref{rel:expand-def} is by exanding the definition of $\alpha_t$ and the final two inequalities are by the definition of an e-value from \eqref{eq:joint-evalue} and $\sum_i \gamma_i \leq 1$. FDR control of \ULOND\  can be proven in a similar fashion by replacing $E_i$ with $S_{\alpha_i^\ELOND}(E_i)$, since $\expect[S_{\alpha_i^\ELOND}(E_i) \mid E_i] = E_i$. Thus, we know that $S_{\alpha_i^\ELOND}(E_i)$ is also an e-value as defined in \eqref{eq:joint-evalue} by the tower property of conditional expectation, and the rest of the proof follows.

 \section{Conclusion}\label{sec:Conclusion}
E-LOND\ and \ULOND\ are two novel procedures that utilize e-values to provide state-of-the-art performance, both practically and theoretically, in power while ensuring provable FDR control under arbitrary dependence. We also built on recent results in using randomization for multiple testing to develop the more powerful randomized online multiple testing procedures of \ULOND\ and \UrLOND. One natural direction is to extend our results to the LORD family of algorithms, which are more powerful, but assign test levels based on the number of hypotheses between the current hypothesis and each of the previous rejections -- more careful analysis is required to ensure FDR control. Note that the sharpness result in \Cref{sec:sharp-fdr} does not preclude this possibility because it only shows that the FDR e-LOND is tight in one specific instance, but e-LOND could be improved in other instances (e.g., have larger test levels when at least one discovery is made). Current LORD algorithms rely on independence and PRDS assumptions to have FDR control while retaining power. Another direction is to explore how e-values can be incorporated with the adaptive online FDR controlling procedures of SAFFRON \citep{ramdas2018saffron} and ADDIS \citep{tian2019addis}, which estimate the proportion of nulls in the manner of Storey-BH \citep{storey_false_discovery_2002}.

\paragraph{Acknowledgements} The authors acknowledge support from NSF grant DMS-1916320.
 \bibliography{hypothesis}
\appendix
\section{Comments on the online multiple testing problem}

We provide additional comments on the motivation behind the formlation of the online multiple testing problem in this section by discussing why FDR is our target error metric and the relationship between online multiple testing and adaptive data analysis.

\subsection{Additional remarks on online FDR control}
\label{sec:WhyFDR}
One might wonder why we wish to simply ensure FDR control, and not prove guarantees about the power of our algorithms as well, e.g., the expected proportion of non-null hypotheses that we actually discover with our algorithm. This is because the in scientific discovery, we cannot know the exact distribution of the statistic under the true distribution when the null hypothesis is false---that would defeat the purpose of testing if the null hypothesis is true in the first place. Prior knowledge or assumptions about the distribution of the true distribution when the null hypothesis is false is often already incorporated by the scientist when designing the individual statistics that are passed to the online multiple testing algorithm. Hence, our framework for online FDR control allows for the user to flexibily change $\alpha_t$ to be large or small based on what they expect the signal of the hypothesis to be.

\subsection{Relating online multiple testing and adaptive data analysis}
\label{sec:ADA}
There is a rich literature on \emph{adaptive data analysis} \citep{dwork_preserving_statistical_2015} that explicitly tackles the data reuse problem, but it is orthogonal to our setup as it focused on the problem of estimation, makes assumptions about the statistic (e.g., bounded) being tested, and focuses on the relation between the number of adaptively chosen parameters can be accurately estimated and the number of i.i.d.\ samples that have been gathered. On the other hand, online multiple testing is agnostic to the exact data generating mechanism (e.g., single dataset, data gathered in a correlated fashion, datasets being merged together, etc.), assumes access to the data only through a statistic (i.e., p-value, e-value, or CI), and maintains error control for a potentially infinite stream of hypotheses, which are not assumed to be adaptively or adversarially chosen. Hence, these two approaches are complementary to each other --- adaptive data analysis focuses on what is the max number of parameters one can estimate for a fixed set of data, while online multiple testing aims to ensure Type I error control regardless of the underlying data sampling method used to test each hypothesis.

\section{Simulation details}
\label{sec:SimulationDetails}

We provide the details of our simulations (in \Cref{sec:Simulations}) in this section. In this section, any references to discount sequence $(\gamma_t)$ is referring to the same choice of $(\gamma_t)$ used in the corresponding algorithm (i.e., \ELOND, \ULOND, \rLOND, or \UrLOND) that is acting on the e-values or p-values. In all our simulations, we let $\gamma_t = 1 / (t(t + 1))$.  We ran the simulations on a 12 core, 60GB RAM cloud server.

\subsection{Definition of \LORD$^*$}\label{sec:lord-star}

We recall the \LORD$^*$ algorithm of \citet{zrnic_asynchronous_2021} as follows:
\begin{align}
    \alpha^{\LORD^*}_t \coloneqq \alpha \left(w_0\gamma_t + \ind{|\rejset_{t - 1}| \geq 1, 1 \not\in \Ccal_t}(\alpha - w_0)\gamma_{t - r_1} + \sum\limits_{i \in \rejset_{t - 1} \setminus [1], i \not\in \Ccal_t} \gamma_{t - i} \right).
\end{align}
Here $w_0 \in [0, \alpha]$ is an algorithm parameter --- we set $w_0 = 0.9$ in all our simulations. $r_1$ is the index of the first discovery made by \LORD$^*$. $(\Ccal_t)$ are a sequence of ``conflict sets'' that dictate hypothesis indices that the current hypothesis has dependence or ``conflict'' with. In our local dependence setting, $\Ccal_t = \{t - L, \dots, t - 1\}$.

\subsection{Local dependence simulation details}\label{sec:local-dep-details}
Each $X_t^i$ is a sample from Beta$(a, b)$ distribution, where we let $a + b = 10^{-2}$, that is shifted and rescaled to be supported on $[-4, 4]$.
The following Hoeffding-based process $(M^i_t)_i$ was shown by \citet{waudby-smith_estimating_means_2023} to be an e-process for random variables bounded in $[\ell, u]$ if $\expect[X^i_t] = 0$ for $ i \in [N]$.
\begin{align}
    M^i_t = \exp\left(\sum\limits_{j = 1}^i\lambda^j_t X^j_t - \frac{(\lambda^j_t(u - \ell))^2}{8} \right),
\end{align}
for any sequence of $(\lambda^j_t)_{j \in [N]}$ that is predictable, i.e., $\lambda_t^j$ can be determined by $X_t^1, \dots,X_t^{j - 1}$. We let $\lambda_t^j = \sqrt{8\log(1 / (\alpha \gamma_t)) / ((u - \ell)^2N)}$ as per \citet[eq. 3.6]{waudby-smith_estimating_means_2023}.

Our e-values, and p-values are defined as follows:
\begin{align}
    E_t = M^{\tau_t}_t\text{ and } P_t = \frac{1}{\max_{i \leq N}\ M^{i}_t},
\end{align} The stopping time $\tau^E_t$ defined the in the following recursive fashion:
\begin{align}
    \tau_t = \min\{i \in [N]: M^i_t \geq 1 / \widehat{\alpha}_t^\ELOND(i)\} \cup \{N\},
\end{align} where we define $\widehat{\alpha}^\ELOND_t(i)$ to be the test level output by \ELOND\ after being applied to $(M_1^{\tau_1 \wedge i}, \dots,  M_{t - 1}^{\tau_{t - 1} \wedge i})$, where $\wedge$ denotes minimum. Note that $(M_1^{\tau_1 \wedge i}, \dots,  M_{t - 1}^{\tau_{t - 1} \wedge i})$ can be computed using only the first $i$ samples of the data for the first $t - 1$ hypotheses, i.e., $\{X_k^j\}_{j \in [i], k \in [t - 1]}$. Hence, these are valid stopping times.

\subsection{Sampling WoR simulation details}\label{sec:sampling-wor-details}
Let $[\ell, u]$ be the support of the population, and in our case, we set $\ell = -4, u = 4$. Let $P(\mu)$ be the distribution $X = (u - \ell) Y + u$ , where $Y \sim \text{Beta}((\mu - \ell) \cdot s / (u - \ell), (u - \mu) \cdot s / (u - \ell))$, i.e., $P(\mu)$
is the Beta distribution scaled to be supported on $[\ell, u]$ with mean $\mu$, and variance scaling factor $s$ (where a smaller $s$ results in population values concentrating at the support limits).
Next, take a discrete grid of size $N \times T$ that is uniformly spread over $[0, 1]$, and compute the quantiles of the grid values of $P(\mu)$. We then shift all quantile values below (or above) $\mu$ by the same amount, so the mean of the grid quantiles is equal to $\mu$.

The e-values and p-values we use in this setup are derived from the following e-process from \citet{waudby-smith_confidence_sequences_2020} for sampling WoR:
\begin{align}
    M^i_t = \exp\left(\sum\limits_{j = 1}^i\lambda^j_t X^j_t  + \mu^{j - 1}_t(0) - \frac{(\lambda^j_t(u - \ell))^2}{8} \right),
\end{align} for any predictable sequence $(\lambda^i_t)_{i \in [N]}$ where $\mu^i_t(0) = \frac{1}{N - i + 1}\sum_{j = 1}^i X_t^j$ is an adjustment term for sampling WoR. We also set $\lambda_t^j = \sqrt{8\log(1 / (\alpha \gamma_t)) / ((u - \ell)^2N)}$ here. We define our e-values and p-values likewise:
\begin{align}
    E_t = M^{\tau_t}_t\text{ and } P_t = \frac{1}{\max_{i \leq N}\ M^{i}_t},
\end{align}
where $\tau_t = \min \{i \in [N]: M^i_t \geq 1 / (\alpha \gamma_t) \} \cup \{N\}$ is the first time the $(M_t^i)_{i \in [N]}$ crosses the threshold $1 / (\alpha \gamma_t)$ or reaches the maximum sample size $N$.

\section{FDR control of e-LOND is sharp}\label{sec:sharp-fdr}

Here we show that there exists a sequence of e-values $(E_1, \dots, E_t)$ such that the FDR control of \ELOND\ is sharp.
\begin{theorem}
    If the discount sequence $(\gamma_t)$ satisfies $\sum_{t \in \naturals} \gamma_t = 1$, there exists a joint distribution over a sequence of e-values $(E_t)_{t \in \naturals}$ such that for every $\varepsilon > 0$, there exists $t' \in \naturals$ such that $\FDR(\rejset^{\ELOND}_{t}) > \alpha - \varepsilon$ for all $t \geq t'$.
\end{theorem}
\begin{proof}
    We write $\rejset_t$ as shorthand for $\rejset_t^\ELOND$.
    We let null be true at every hypothesis, i.e., $\hypset_0 = \naturals$, and construct the joint distribution over e-values is characterized as follows:
    \begin{gather}
        \xi_t \coloneqq \{E_t = (\alpha \gamma_t)^{-1} \text{ and }E_i = 0\text{ for all }i \neq t\}, \qquad \xi_0 \coloneqq \{E_t = 0 \text{ for all }t \in \naturals\}\\
        \prob{\xi_t} = \alpha \gamma_t \text{ for each }t \in \naturals,\qquad  \prob{\xi_0} = 1 - \alpha.
    \end{gather}
Note that $\xi_t$ are disjoint events for $t \in \naturals \cup \{0\}$, and $\prob{\xi_0}+\sum_{t \in \naturals}\prob{\xi_t} = 1 - \alpha  \alpha \sum_{t \in \naturals}\gamma_t.$ --- hence this characterizes a complete distribution over $(E_t)_{t \in \naturals}$. Further, $\expect[E_t] = (\alpha \gamma_t)^{-1} \cdot \prob{\xi_t} = 1$, for each $t \in \naturals$, so $(E_t)_{t \in \naturals}$ is provably a sequence of e-values.

We note that $\FDP(\rejset_{t}) = \max_{t_1 \in [t]}\ind{\xi_{t_1}} $, i.e., the FDP is 1 iff $\xi_{t_1}$ for some $t_1 \in [t]$ occurs.
Hence,
\begin{align}
    \FDR(\rejset_{t}) = \expect\left[\max_{t_1 \in [t]}\ind{\xi_{t_1}}\right] = \prob{\bigcup_{t_1 \in [t]} \xi_{t_1}} = \sum\limits_{t_1 \in [t]}\prob{\xi_{t_1}} = \alpha \sum\limits_{t_1 \in [t]} \gamma_{t_1}.\label{eq:counterex-fdr}
\end{align}

Hence, for a fixed $\varepsilon > 0$, if we define $t'(\varepsilon)$ to be the smallest $t \in \naturals$ such that $\sum_{t_1 \in [t]} \gamma_{t_1} > 1 - (\varepsilon / \alpha)$ --- note such a $t$ always exists because $(\gamma_t)$ is nonnegative and $\sum_{t \in \naturals} \gamma_t = 1$. We can see as a result of \eqref{eq:counterex-fdr}, $\FDR(\rejset_{t}) > \alpha - \epsilon$  for all $t \geq t'(\epsilon)$. Thus, we have shown our desired result.
\end{proof}

A similar argument can be made to argue that \ULOND\ is sharp as well, as well as FCR control of \ELOND-CI\ and \ULOND-CI.

\end{document}